\documentclass[11pt]{article}
\usepackage{amsfonts}
\usepackage{jfe}          

\usepackage{bm}

\usepackage[title]{appendix}
\usepackage{multirow}
\usepackage{graphicx}
\usepackage{array}
\usepackage{color}
\usepackage{mathtools,mathrsfs}
\usepackage{tabularx,tabulary}
\usepackage{verbatim}
\usepackage{dsfont}
\usepackage{xcolor,colortbl}
\usepackage[sort]{natbib}
\usepackage[T1]{fontenc}
\usepackage{eurosym} 
\usepackage{lmodern}
\usepackage[normalem]{ulem}
\usepackage{bbm}
\usepackage{hyperref} 
\DeclareGraphicsExtensions{.jpg,.eps,.pdf} 
\usepackage{enumerate}
\usepackage{setspace}
\usepackage{amsmath,amssymb,amsfonts,amsthm} 
\newtheorem{theorem}{Theorem}

\newtheorem{definition}[theorem]{Definition}

\newtheorem{lemma}[theorem]{Lemma}

\newtheorem{proposition}[theorem]{Proposition}

\theoremstyle{remark}

\definecolor{bleu}{cmyk}{1,0.4,0,0}

\newcommand{\E}{\mathbb{E}}

\definecolor{Gray}{gray}{0.95}
\definecolor{White}{gray}{1.00}

\usepackage{xcolor}
\hypersetup{
citebordercolor=white,
filebordercolor=red,
linkbordercolor=blue
}
\hypersetup{
colorlinks,
citecolor=blue,
linkcolor=red,
urlcolor=blue}

\newcommand{\1}{\mathds{1}}
\newcolumntype{a}{>{\columncolor{Gray}}c}
\newcommand{\mtrx}[1]{\bm{#1}}
\DeclarePairedDelimiter{\abs}{\lvert}{\rvert}
\DeclarePairedDelimiter{\norm}{\lVert}{\rVert}

\begin{document}


\title{Dirichlet policies for reinforced factor portfolios}

\author{Eric Andr\'e\thanks{EMLYON Business School, 23 avenue Guy de Collongue, 69130 Ecully, FRANCE. E-mail: eandre@em-lyon.com} \and Guillaume Coqueret\thanks{EMLYON Business School, 23 avenue Guy de Collongue, 69130 Ecully, FRANCE. E-mail: coqueret@em-lyon.com}}
\maketitle
\vspace{-6mm}
\begin{center}
\large
\end{center}
\begin{abstract}
This article aims to combine factor investing and reinforcement learning (RL). The agent learns through sequential random allocations which rely on firms' characteristics. Using Dirichlet distributions as the driving policy, we derive closed forms for the policy gradients and analytical properties of the performance measure. This enables the implementation of REINFORCE methods, which we perform on a large dataset of US equities. Across a large range of parametric choices, our result indicates that RL-based portfolios are very close to the equally-weighted (1/$N$) allocation. This implies that the agent learns to be \textit{agnostic} with regard to factors, which can partly be explained by cross-sectional regressions showing a strong time variation in the relationship between returns and firm characteristics.

\end{abstract}
\textbf{Keywords}: Reinforcement learning; Factor investing; Equally-weighted portfolio; Asset pricing.  
\textbf{JEL classifications}: C38; G11; G12


\section{Introduction}

The traditional econometric approaches to asset pricing have recently seen a surge in competition from machine learning tools. A flow of recent studies\footnote{See for instance \cite{chen2019deep}, \cite{feng2019deep}, \cite{gu2020autoencoder} and \cite{gu2020empirical}.} have shown the benefits that can be reaped when switching from the conventional linear models to more complex structures such as tree methods or neural networks. Supervised learning algorithms help the econometrician link financial performance (asset returns) to key indicators such as firm characteristics (\cite{gu2020empirical}) or latent factors (\cite{kelly2019characteristics}, \cite{lettau2020estimating, lettau2020factors}). Based on large datasets, these black boxes reveal intricate correlations between variables that are not captured by standard linear models, thereby often improving cross-sectional fit. Depending on the quality of the sample and the algorithm's architecture, these correlations may nonetheless be spurious. They are likely to hold out-of-sample only if they reflect genuine causality relationships, which are much harder to uncover.\footnote{We refer to \cite{pearl2009causality} for an exhaustive treatment on causal models and to \cite{arjovsky2019invariant} and \cite{pfister2019invariant} for recent perspectives on causality in machine learning models. } 

Beyond supervised learning, researchers have resorted to another powerful family of techniques to understand and predict returns: reinforcement learning (RL). Contributions range from early tests of \cite{neuneier1996optimal} and \cite{moody1998performance} to the more recent work of \cite{deng2016deep}, \cite{li2019deep} and \cite{wang2020continuous}.\footnote{These references are by no means an exhaustive account of the literature on this subject. On the arXiv repository only, more than 20 papers including the terms "reinforcement learning" in their title have been posted in the \textit{quantitative finance} (qfin) section in 2019 only. We also direct to the survey \cite{sato2019model} for more references on RL applied to portfolio optimization.} Two common threads between these studies is that they often originate from the field of computer science and that they work with price data only (at high frequencies most of the time). To the best of our knowledge, there are no contributions that seek to harvest the information contained in firm specific attributes and combine it with reinforcement routines to produce factor-based portfolios. One goal of the present paper is to fill this void.

The main challenge when implementing RL algorithms for the purpose of trading is the modelling of the environment. The infinite dimensions of the state space (firm attributes) and action space (investment policies) make many approaches relying on Markov decision processes (MPD) inadequate. This is because a focal tool in MDP analysis is the value function, which measures the expected gain or reward for any given action or state. In the framework of factor investing, these states and actions cannot be properly discretized without either making overly simplistic assumptions, or rendering the computations intractable.

In order to bypass these technical hurdles, one solution is to resort to the so-called policy gradient approach. In this case, the decisions are made according to a parametric function which probabilistically determines which actions (i.e., investments) to perform. The agent then learns by sequentially updating the policy parameters after receiving flows of rewards (e.g., returns). Most of the time, the policy is modelled by neural networks (NNs), which is a convenient choice, given their flexibility. It is for instance the option chosen by \cite{deng2016deep} and \cite{zhang2019deep}. One drawback of general purpose NNs is that their output cannot be directly translated into portfolio weights, because it violates the budget constraint. The core idea of the present paper is to resort to a special class of distributions that circumvent this issue by directly yielding the investment allocations. 

Indeed, Dirichlet distributions have the opportune property of being defined on simplexes, which makes them appropriate to model long-only portfolio compositions. In fact, Dirichlet distributions have already been used in related studies. \cite{cover1996universal} find that two such distributions yield portfolio allocations with interesting theoretical properties. More recently, \cite{le2019portfolio} rely on Dirichlet distributions to derive robust estimates of the efficient frontier, and \cite{korsos2013dirichlet} uses them to estimate the composition of hedge fund portfolio holdings. In a similar vein, \cite{sosnovskiy2015financial} shows that Dirichlet laws can be used to approximate the distribution of stock weights in aggregate market indices. 

One of the simple but novel contributions of the paper is to link the RL policy to firm-specific attributes. To this purpose, the inspiration comes from earlier work on characteristics-based investing.\footnote{See, e.g., \cite{haugen1996commonality}, \cite{daniel1997evidence}, \cite{brandt2009parametric}, \cite{hjalmarsson2012characteristic} and \cite{ammann2016characteristics}). Our approach is closer in spirit to the most recent of these references.} The idea is to map a linear combination of the characteristics into portfolio weights. While the traditional models aim to optimize expected utility functions, our approach seeks to maximize expected gains. The simplest definition of gain is a portfolio return but it is possible to adjust it to risk via the sequential Sharpe ratio computations presented in \cite{moody1998performance}. 

Our contribution is threefold. First, we propose a tractable formulation of the reinforcement learning problem when designing portfolio allocations based on firm specific attributes. To the best of our knowledge our approach is the first to articulate the combination between factor investing and RL in such a simple fashion. Second, we employ our methodology on a large dataset of US equities. Our results are qualitatively homogeneous, despite the numerous degrees of freedom in the implementation, and they indicate that the agent should be better of by \textit{ignoring} the informational content provided by firm-specific attributes. Finally, we provide two attempts to explain our results. One direction is linked to the pricing ability of characteristics, which, we find, is quantitatively weak. The second angle stems from an asset pricing model in which the noise of loadings plays a crucial role. We compare the reinforcement learning process to a simple factor-based quadratic optimization. The two are hard to reconcile, except for one salient stylized fact: both methods recognize a strong common factor within the cross-section of stock returns. Consequently, portfolios allocate almost uniformly across assets, except after major market crashes.

The paper is structured as follows. In Section \ref{sec:1}, we lay out the theoretical foundations of RL-based factor investing. Section \ref{sec:2} is dedicated to a detailed presentation of the dataset and the implementation protocol. Our empirical results are outlined in Section \ref{sec:3}. Section \ref{sec:4} provides explanatory perspectives based on the pricing ability of firm attributes and on a simple asset pricing model. Finally, \ref{sec:5} concludes.

\section{Reinforcement learning meets factor investing}
\label{sec:1}
This section is dedicated to the presentation of all concepts and theoretical apparatus developed and required in the paper. 

\subsection{The framework}
We study a dynamic discrete time investment problem with finite horizon $T$. The investable universe consists of $N$ assets indexed by $n=1,\dots, N$. There are $K$ characteristics associated to each asset. We refer to section \ref{sec:data} for a list of those retained in the empirical section of this study. To allow for a bias or non-zero intercept in our model, we add a constant characteristic equal to 1. Therefore, at time $t \in \{0,1,\dots,T\}$, asset $n$ is described by a $(K+1)$-dimensional vector $(\mtrx{x}_{t,n})^\intercal = [x_{t,n}^{(0)}, \dotsc, x_{t,n}^{(k)}, \dotsc, x_{t,n}^{(K)}]$ where $x_{t,n}^{(0)}=1$ is an indicator that is kept fixed through the cross-section of assets.

Among these characteristics are $p_{t,n}$, the time-$t$ price of asset $n$, and $d_{t,n}$, the dividend per share issued between time $t-1$ and $t$. The total return of asset $n$ between $t-1$ and $t$ is therefore $r_{t,n} = (p_{t,n}+d_{t,n})/p_{t-1,n}-1$. In our setting, we can work with price returns (omitting dividends) or total returns interchangeably, as they are simply two different drivers of rewards for the investor.

We use the bold notations $\mtrx{r}_{t}$ for the vector of the returns of all assets at time $t$ and $\mtrx{X}_t$ for the $N\times (K+1)$ matrix of characteristics at time $t$ whose $n$-th row is $(\mtrx{x}_{t,n})^\intercal$. We denote by $\mathcal{M}$ the set of these $N\times (K+1)$ matrices whose first column is $\1$, the vector of 1.

The agent posits a factor model for the returns of the assets
\begin{equation}\label{eq:FactorModel}
\mtrx{r}_{t+1} = f(\mtrx{X}_{t}) + \mtrx{\epsilon}_{t+1}
\end{equation}
where $f$ is a function from $\mathcal{M}$ to $\mathbb{R}^{N}$ and $\mtrx{\epsilon}_{t}$ is an i.i.d. White Noise with mean vector equal to 0 and a diagonal  correlation matrix $\mtrx{\Sigma}_{\epsilon}$. The diagonal elements of $\mtrx{\Sigma}_{\epsilon}$ are $\sigma_{n}^{2}$, the idiosyncratic variances of the assets. Let $\bm{P}_\epsilon$ denote the law of the r.v. $\mtrx{\epsilon}_{t}$, defined on $\mathbb{R}^{N}$.

A standard assumption in the finance literature is that the function $f$ is a linear map that can be represented by a $(K+1)$ vector $\mtrx{\beta}$, that is $f(\mtrx{X}_{t}) = \mtrx{X}_{t} \mtrx{\beta}$. Note however that we do not need this assumption in our study.

\subsection{Markov Decision Process}
We assume that the investment problem can be formulated as a finite horizon Markov Decision Process (MDP). At each time $t$, the agent observes the state $S_t$ of the system (the characteristics of the investable universe and of her portfolio) and then takes an action $A_t$ (a choice of a composition for her portfolio). Finally, the agent obtains a time-($t+1$) reward, which is linked to the return of her portfolio between $t$ and $t+1$ and the system transition to the next state. We now describe formally this MDP.\footnote{See, e.g., \citet{bauerle2011markov}.}

\paragraph{Actions.}
The \emph{action} $A_{t}$ taken by the agent at time $t$ is the choice of a vector $\mtrx{w}_{t} \in \mathbb{R}^{N}$, which is the composition of her portfolio. We consider the case where there is no short selling. The restriction to positive weights is realistic since most asset managers have long-only constraints. This is typically the case of institutional investors (see \citealp{koijen2019demand}). Therefore, $\mtrx{w}_{t}$ must be in the $N-1$ simplex $\Delta$ (we omit the dimension superscript to lighten notations), which is then the \emph{action space}:
\begin{equation}
\Delta = \left\{ \left(w_{1},\dotsc,w_{N}\right)\in\mathbb{R}^{N} \colon \sum_{n=1}^{N}w_{n} = 1 \textrm{ and } w_{n} \geq 0 \textrm{ for all } n \right\}.  \label{eq:simplex}
\end{equation}
Seen as a subset of $\mathbb{R}^{N-1}$, it is endowed with the inherited Borel $\sigma$-algebra that we denote $\mathscr{B}(\Delta)$ .

\paragraph{Rewards.}
The agent's objective is to maximise her utility of, or some performance measure of, the terminal value of her portfolio $V_T = V_0 + \sum_{t=0}^{T-1} V_{t} \rho_{t+1}$, where $\rho_{t+1} \coloneqq \mtrx{w}_{t}^\intercal \mtrx{r}_{t+1}$ is the return of her portfolio between $t$ and $t+1$. We will consider two cases: a \emph{risk insensitive} agent who seeks to maximize her profit and a \emph{risk sensitive} agent whose goal is to maximize the differential Sharpe Ratio proposed by \citet{moody1998performance}.

In the first case, the agent's reward at time $t$ for the action taken at time $t-1$ is simply
\begin{equation}
R_{t} = \rho_{t}
\end{equation}
In the second case, it is
\begin{equation} \label{eq:SR}
R_t \coloneqq \mathrm{SR}_t=\frac{\hat\mu_t}{K_\kappa\sqrt{\hat\sigma^2_t-\hat\mu_t^2}}
\end{equation}
with
\begin{equation*}
\begin{array}{lll}
\hat\mu_t = \kappa \rho_{t} + (1-\kappa) \hat\mu_{t-1} & \text{exponentially weighted (EW) moving average of returns}; \\
\hat\sigma^2_t = \kappa \rho_{t}^2 + (1-\kappa) \hat\sigma^2_{t-1} & \text{EW moving average of squared returns}; \\
K_\kappa =\sqrt{\frac{1-\kappa/2}{1-\kappa}} & \text{scaling factor}.
\end{array}
\end{equation*}

We draw attention to the use of lowercase $r_t$ for individual asset returns, uppercase $R_t$ for rewards, and $\rho_{t}$ for portfolio returns. These are closely linked, but not equal.

\paragraph{States.}
As we want the agent to choose an action using the asset characteristics, $\mtrx{X}_{t}$ must be included in the \emph{state}. The chosen reward defines which data must be added to the state: for the risk-insensitive agent $\rho_{t}$ is enough, for the differential Sharpe Ratio, we should also include $\hat\mu_t$ and $\hat\sigma^2_t$. In the latter case, it is still $\rho_{t}$ that drives deterministically the additional data, therefore we will henceforth consider without loss of generality the case where a state corresponds to the couple
\begin{equation*}
S_{t} = (\rho_t, \mtrx{X}_{t})
\end{equation*}

The \emph{state space} is then $\mathcal{S} = \mathbb{R} \times \mathcal{M}$. The set $\mathcal{M}$, is a subset of the space of the $N\times (K+1)$ matrices which can be identified with $\mathbb{R}^{N\times (K+1)}$. It is therefore endowed with the inherited Borel $\sigma$-algebra that we denote $\mathscr{B}(\mathcal{M})$. The state space itself is endowed with the product $\sigma$-algebra $\mathscr{B}(\mathcal{S}) = \mathscr{B}(\mathbb{R}) \otimes \mathscr{B}(\mathcal{M})$.

\paragraph{Episodes.}
The agent having observed the state of the system takes an action, then the system transitions to the next state and the agent receives a reward. A given realization of this interaction between the agent and her environment is an \emph{episode}:
\begin{equation*}
S_{0},A_{0},R_{1},S_{1},A_{1},R_{2},S_{2},A_{2},\dotsc,S_{T-1},A_{T-1},R_{T},S_{T}
\end{equation*}
At any date $t$, the \emph{cumulative discounted return} can be computed. It is the sum of the future rewards in this episode, possibly discounted at a discount rate $0 < \gamma \le 1$:
\begin{equation*}
G_{t} = \sum_{l=1}^{T-t} \gamma^{l-1} R_{t+l} = R_{t+1} + \gamma G_{t+1}.
\end{equation*}

\paragraph{Transition probability.}
How the system transitions to the next state $S_{t+1}$ given some previous state $S_{t}$ and action $A_{t}$ is given by the \emph{state transition probability}
\begin{equation*}
\mathrm{Prob}\left( S_{t+1} \in B \mid S_{t}, A_{t} \right), \quad B \in \mathscr{B}(\mathcal{S}).
\end{equation*}
We assume that the matrix of asset characteristics is a Markov process whose evolution is driven by the transition probabilities
\begin{equation*}
\mathbb{P}_{t}^{u} \left( M \mid \mtrx{X} \right) = \mathrm{Prob} \left( \mtrx{X}_{u} \in M \mid \mtrx{X}_t=\mtrx{X} \right), \qquad u>t,
\end{equation*}
which are independent of the value of the portfolio and of the choice of the action. Therefore, if $B$ is a Cartesian product of Borel sets, $B = C \times M$, where $C \in \mathscr{B}(\mathbb{R})$ and $M \in \mathscr{B}(\mathcal{M})$, we obtain the factorization
\begin{equation*}
\mathrm{Prob}\left( S_{t+1} \in B \mid S_{t}, A_{t} \right) = \mathrm{Prob}\left( \rho_{t+1} \in C \mid S_{t}, A_{t} \right) \mathbb{P}_{t}^{t+1} \left( M \mid \mtrx{X}_t \right).
\end{equation*}
In our specific setting with a factor model, we have a transition function $\mathbb{T}$ that gives the next value of $\rho_{t+1} $ given the state and action at $t$. This is $\rho_{t+1} = \mathbb{T}(\mtrx{X}_{t}, \mtrx{w}_{t}, \mtrx{\epsilon}_{t+1}) = \mtrx{w}_{t}^\intercal (f(\mtrx{X}_{t}) + \mtrx{\epsilon}_{t+1})$. When $S_{t}$ and $A_{t}$ are known, the value of $\rho_{t+1}$ is driven by $\mtrx{\epsilon}_{t+1}$ and conversely, if $r \in \mathbb{R}$, then $\mathbb{T}^{-1}(r \mid \mtrx{X}_t, \mtrx{w}_{t})$ is the hyperplane orthogonal to the vector $\mtrx{w}_{t}$ translated by the vector $r \1 - f(\mtrx{X}_{t})$. Finally, we can write
\begin{equation}\label{eq:TransitionProba}
\mathrm{Prob}\left( S_{t+1} \in C \times M \mid S_{t}, A_{t} \right) = \bm{P}_\epsilon ( \mtrx{\epsilon}_{t+1} \in \mathbb{T}^{-1}(C \mid \mtrx{X}_t, \mtrx{w}_{t})) \mathbb{P}_{t}^{t+1} \left( M \mid \mtrx{X}_t \right)
\end{equation}

\subsection{Policy}
To allow for the exploration of all actions versus the exploitation of the optimal action, we will use a stochastic policy that gives the probability of choosing an action $A_{t}$ given the state $S_{t}$. Specifically, we will study policies $\pi(\cdot \mid S_{t}, {\mtrx{\theta}})$ defined on $\mathscr{B}(\Delta)$ with parameter $\mtrx{\theta}=[\theta^{(1)},\dots, \theta^{(K)}]^\intercal$. At each time step, we will draw from this distribution to select an action. More precisely, we are looking in this study for a policy that only takes into account the asset characteristics, hence we restrict ourselves to policies that takes the form
\begin{equation*}
A_{t} = \mtrx{w}_{t} \sim \pi \left( \cdot \mid \mtrx{X}_{t}, \mtrx{\theta} \right) .
\end{equation*}

We will use the shorthand notations $\pi_{\mtrx{\theta}}$ for $\pi \left( \cdot \mid \mtrx{\theta} \right)$ and $\mathbb{E}_{\mtrx{\theta}}[\cdot]$ or $\mathbb{E}_{\pi}[\cdot \mid \mtrx{\theta}]$ for the expectation under the policy $\pi_{\mtrx{\theta}}$.

\paragraph{Value function.}
The \emph{value function} at $t$ of the state $S_{t}$ under the policy $\pi_{\mtrx{\theta}}$, is the expected value of the cumulative discounted return from $t$ onward, when this policy is chosen to select the actions at each future time steps:
\[ V^{\mtrx{\theta}}(t,S_{t}) = \mathbb{E}_{\mtrx{\theta}} \left[ G_{t} \mid S_{t} \right] = \sum_{l=1}^{T-t} \gamma^{l-1} \mathbb{E}_{\mtrx{\theta}}\left[ R_{t+l} \mid S_{t} \right] .\]

To find the optimal policy, the standard tool is dynamic programming, for which the value function must satisfy the recursive Bellman equation (see Chapter 4 in \citealp{sutton2018reinforcement}). However, for the differential Sharpe Ratio, it is known that the introduction of the variance in the reward renders the problem time-inconsistent. In this paper, we will use RL algorithms to explore the optimal policies. Nonetheless, in the case of the risk insensitive agent, the problem can also be solved with dynamic programming as the next result shows.

\begin{proposition}\label{prop:PolicyValue}
For the risk insensitive agent, the time $t$ expected values of the future rewards are given by
\begin{equation*}
\mathbb{E}_{\mtrx{\theta}} \left[ R_{t+l}\mid S_{t}=(\rho_t, \mtrx{X}_{t}) \right] =
\begin{cases}
\mathbb{E}_{\mtrx{\theta}} \left[ \mtrx{w}_{t}\mid \mtrx{X}_{t} \right]^\intercal f\left(\mtrx{X}_{t}\right) &l=1,\\
\int_{\mathcal{M}} \mathbb{E}_{\mtrx{\theta}} \left[ \mtrx{w} \mid \xi \right]^\intercal f\left(\xi\right)\mathbb{P}_{t}^{t+l-1}(d\xi\mid\mtrx{X}_{t}) &l \ge 2.
\end{cases}
\end{equation*}

The policy value satisfies the recursive Bellman equation 
\begin{align*}
V^{\mtrx{\theta}}(t,\mtrx{X}_{t}) &= \mathbb{E}_{\mtrx{\theta}} \left[ R_{t+1} \mid S_{t} \right] + \int_{\mathcal{M}} V^{\mtrx{\theta}}(t+1,\xi) \mathbb{P}_{t}^{t+1}(d \xi \mid \mtrx{X}_{t}) ,\\
V^{\mtrx{\theta}}(T-1,\mtrx{X}_{T-1}) &=\mathbb{E}_{\mtrx{\theta}} \left[ R_{T} \mid S_{T-1} \right].
\end{align*}
\end{proposition}
\begin{proof}
See Appendix \ref{sec:Proof-PolicyValue}.
\end{proof}

\paragraph{Performance measure.}
The \emph{performance measure} of the policy is its value from some initial state $S_{0}$: $J(\mtrx{\theta})=V^{\mtrx{\theta}}(0,S_{0})$. Our aim is to find a parameter of the policy that maximizes this performance measure
\begin{equation}\label{eq:OptPolicy}
\mtrx{\theta}^{\ast} \in \mathop{\mathrm{arg\,max}}_{\mtrx{\theta}}J(\mtrx{\theta}).
\end{equation}
Before taking on this task, we specify the parametrized form of the policy that we use in this study.

\subsection{Dirichlet policies}
\label{sec:diripol}
One of the main contribution of this paper is the use of Dirichlet distributions to define the policy of the agent. We find it particularly well suited for describing portfolio weights when short selling is proscribed. We first briefly recall its definition and some of its properties that are used thereafter.
 
\paragraph{Definition.}
The Dirichlet distribution is defined on the $N-1$ simplex $\Delta$ and its density is zero outside $\Delta$. It is parametrized by a vector $\mtrx{a} = [ a_{1}, a_{2}, \cdots, a_{N} ]^\intercal$ of \emph{concentration parameters} where $a_{n}>0$ for all $n=1,\dots,N$. We will use the notation $\sigma$ for the \emph{scale parameter}
\[ \sigma=\sum_{n=1}^{N}a_{n}=\1^\intercal\mtrx{a}.\]

The probability density function (pdf) is given by
\begin{equation}\label{eq:dirichlet}
f(w_{1},\dotsc,w_{N}\mid\mtrx{a})=\frac{1}{B(\mtrx{a})}\prod_{n=1}^{N}w_{n}^{a_{n}-1} ,
\end{equation}
where the normalizing constant is the Multivariate Beta function, which can be written with the Gamma function as follows
\[ B(\mtrx{a})=\frac{\prod_{n=1}^{N}\Gamma\left(a_{n}\right)}{\Gamma\left(\sigma\right)}. \]

\paragraph{Some properties.}
Let $\mtrx{w} = [ w_{1}, \cdots, w_{N} ]^\intercal$ be a vector with Dirichlet distribution, which we denote by $\mtrx{w} \sim \mathrm{Dir} \left(\mathbf{a}\right)$. The marginal distributions are Beta distributions: for all $n$,
\[ w_{n} \sim \mathrm{Beta} \left(a_{n},\sigma-a_{n}\right) ,\]
from which we get\footnote{Additional properties required for some proofs in this paper are collected in Appendix~\ref{sec:PropDirichlet}.}
\begin{align}\label{eq:meanDirichlet}
\mathbb{E}\left[w_{n}\right] &= \frac{a_{n}}{\sigma} ,&\mathrm{Var}(w_n) &= \frac{1}{\left(\sigma+1\right)}\left\{ \frac{a_{n}}{\sigma}\left(1-\frac{a_{n}}{\sigma}\right)\right\} .
\end{align}
When $\mtrx{w}$ is the composition of a portfolio, these properties make clear the link between the concentration parameters $\mtrx{a}$ and the average relative shares of each asset in the portfolio and the inverse relationship between the marginal variances and the scale parameter $\sigma$.

\paragraph{Link with the asset characteristics.}
In this paper, we study the policy for which the probability of choosing action $\mtrx{w}_{t}$ at time $t$ has the Dirichlet distribution with concentration parameters $\mathbf{a}_{t}=[a_{t,1} \ a_{t,2} \ \cdots \ a_{t,N}]^\intercal$ where $a_{t,n}>0$ for all $n$. We posit that the concentration parameters are functions of the asset characteristics. Two possible forms are studied:
\begin{equation}\label{eq:policies}
\mathbf{a}_{t}=
\begin{cases}
\mtrx{X}_{t}\mtrx{\theta}_{t} & \left(\mathbf{F1}\right)\\
e^{\mtrx{X}_{t}\mtrx{\theta}_{t}} & \left(\mathbf{F2}\right).
\end{cases}
\end{equation}

The first form is a simple linear combination which is highly tractable, but may violate the condition that $a_{t,n}>0$ for some values of $\theta_{t}^{(k)}$. Indeed, during the learning process, an update in $\mtrx{\theta}$ might yield values that are out of the feasible set of $\mtrx{a}_t$. In this case, it is possible to resort to a trick that is widely used in online learning (see, e.g., Section 2.3.1 in \citealp{hoi2018online}). The idea is simply to find the acceptable solution that is closest to the suggestion from the algorithm. If we call $\mtrx{\theta}^*$ the result of an update rule from a given algorithm, then the closest feasible vector is 
 \begin{equation}
 \mtrx{\theta}= \underset{\mtrx{z} \in \Theta(\mtrx{X}_t)}{\mathop{\mathrm{arg\,max}}} ||\mtrx{\theta}^*-\mtrx{z}||^2,
 \label{eq:proj}
 \end{equation}
where $||\cdot||$ is the Euclidean norm and $\Theta(\mtrx{X}_t)$ is the feasible set, that is, the set of vectors $\mtrx{\theta}$ such that the $a_{t,n}=\theta_{t}^{(0)} + \sum_{k=1}^K \theta_{t}^{(k)}x_{t,n}^{(k)}$ are all nonnegative.

The second form of the policy is slightly more complex but remains always valid.

The combination of the Dirichlet distribution with time-varying weights $\mtrx{w}_t$ and parameters $\mtrx{a}_t$ defined above yields a policy $\pi(\cdot \mid \mtrx{X}_t, \mtrx{\theta}_t)$ that depends on exogenous characteristics $\mtrx{X}_t$ as well as $K+1$ parameters, stacked in the vector $\mtrx{\theta}_t$. By equation~\eqref{eq:meanDirichlet}, under this policy
\begin{align}\label{eq:wgts}
\mathbb{E}_{\pi}\left[w_{t,n}\mid \mtrx{X}_{t},\mtrx{\theta}_{t}\right] &= \frac{a_{t,n}}{\sigma_{t}} &\text{where}
&&a_{t,n} &=
\begin{cases}
(\mtrx{x}_{t,n})^\intercal\mtrx{\theta}_{t} & (\mathbf{F1})\\
e^{(\mtrx{x}_{t,n})^\intercal\mtrx{\theta}_{t}} & (\mathbf{F2})
\end{cases}.
\end{align}

There is a very strong link between this formulation and other methods that link financial performance to firm-specific characteristics like \cite{brandt2009parametric} and \cite{ammann2016characteristics}. One common feature is that for any $k\neq 0$, the parameter $\theta^{(k)}$ synthesizes the impact of feature $k$ on the whole cross-section of returns. If $\theta^{(k)}$ is positive (\textit{resp.},, negative), then, on average, the corresponding feature is expected to have a positive (\textit{resp.},, negative) effect on returns. The parameter $\theta^{(0)}$ is intended to reflect some idiosyncrasy that is not rendered by the characteristics but that is shared by all assets.

In our implementation, the link between asset characteristics and the portfolio weights can be made more explicit. Indeed, in Appendix \ref{sec:link}, we show that, for policy $(\mathbf{F1})$,
\begin{equation}
\mathbb{E}_{\pi}\left[w_{t,n}\mid \mtrx{X}_{t},\mtrx{\theta}_{t}\right] = \frac{1}{N} + \tilde{w}_{t,n} ,
\end{equation}
where $\tilde{w}_{t,n} = \frac{1}{N\theta_{t}^{(0)}}\sum_{k=1}^K \theta_{t}^{(k)}x_{t,n}^{(k)}$ are such that $\sum_{n=1}^N \tilde{w}_{t,n} = 0$. Therefore, policy $(\mathbf{F1})$ can be seen as targeting an investment in the Equally Weighted portfolio and a long-short portfolio whose weights are driven by the assets' characteristics, just as in \cite{brandt2009parametric} and \cite{ammann2016characteristics}. For policy $(\mathbf{F2})$, we show that a similar relationship holds approximately for small values of the sum $\sum_{k=1}^K \abs{\theta_{t}^{(k)}}$.

Finally, we see from this expression of the average weights that, if the RL algorithm is unable to find a long-short portfolio that improves the return of the Equally Weighted portfolio, then it will fall back on the latter. This is briefly discussed in Appendix \ref{sec:link} and explored in details in section~\ref{sec:PAC}.

\subsection{The policy gradient method}
\label{sec:pgm}
The optimization problem \eqref{eq:OptPolicy} cannot be solved by dynamic programming when the reward is the differential Sharpe Ratio. We thus search for an approximate solution using the method named Policy Gradient \citep[Chapter~13]{sutton2018reinforcement}. This method can deal with the infinite state space $\mathcal{S}$ and seeks to learn a parametrized policy by updating the parameter via gradient ascent in $J$:
\begin{equation}
 \mtrx{\theta}_{t+1} = \mtrx{\theta}_{t} + \alpha \widehat{\nabla J(\mtrx{\theta}_t)} ,
\label{eq:pgrad}
\end{equation}
where $\widehat{\nabla J(\mtrx{\theta}_t)}$ is a stochastic estimate of the gradient of the performance measure (with respect to $\bm{\theta}_t$) and $\alpha\in(0,1)$ is a learning rate.

The core result when implementing policy gradient learning is the so-called Policy Gradient Theorem:
\begin{equation}\label{eq:PGT}
\nabla J(\mtrx{\theta}_t) = \mathbb{E}_{\pi} \left[ G_t \nabla\ln\pi\left(\mtrx{w}_{t}\mid S_{t},\mtrx{\theta}_{t}\right)\mid \mtrx{X}_t,\mtrx{\theta}_{t} \right],
\end{equation}
which is incredibly convenient because the two terms in the expectation are disentangled. We refer to Section 13.3 in \cite{sutton2018reinforcement} for a proof of this result. It is thus imperative to derive analytical expressions for $\nabla \ln \pi_{\mtrx{\theta}}$ as per the following Proposition.

\begin{proposition}\label{prop:DirichletGrad}
For a Dirichlet policy, the gradients are given by
\begin{equation}\label{eq:DirichletGrad}
\nabla\ln\pi\left(\mtrx{w}_{t}\mid \mtrx{X}_t,\mtrx{\theta}_{t}\right) = \sum_{n=1}^{N}\left(\digamma\left(\sigma_{t}\right)-\digamma\left(a_{t,n}\right)+\ln w_{t,n}\right)\nabla a_{t,n},
\end{equation}
where
\begin{equation*}
\nabla a_{t,n} = 
\begin{cases}
\mtrx{x}_{t,n} & \left(\mathbf{F1}\right)\\
e^{(\mtrx{x}_{t,n})^\intercal\mtrx{\theta}_{t}} \mtrx{x}_{t,n} & \left(\mathbf{F2}\right)
\end{cases}.
\end{equation*}
\end{proposition}
\begin{proof}
See Appendix \ref{sec:Proof-DirichletGrad}.
\end{proof}

The policy gradient is then the weighted sum over all assets of the concentration parameters' gradients. From equation~\eqref{eq:wgts}, we see that each gradient $\nabla a_{t,n}$ is the direction in parameter space along which the relative importance in the portfolio of asset $n$ increases. The weights are understood using Proposition~\ref{prop:wlnw}. Indeed we have
\begin{equation}\label{eq:logMean}
\digamma\left(\sigma_{t}\right)-\digamma\left(a_{t,n}\right)+\ln w_{t,n}=\ln w_{t,n}-\mathbb{E}_{\pi}\left[\ln w_{t,n}\mid \mtrx{X}_t,\mtrx{\theta}_{t}\right],
\end{equation}
and therefore $\nabla a_{t,n}$ has a positive (negative) weight if the realized log weight of asset $n$ is above (below) its expected value.

Finally, Proposition~\ref{prop:DirichletGrad} sheds light on the learning process implied by the policy gradient method described by equations~\eqref{eq:pgrad} and~\eqref{eq:PGT}. When the cumulative discounted return $G_t$ is positive (\textit{resp.}, negative), assets which had, at time $t$, their realized log weights above their expected values will see their expected weights at time $t+1$ increase (\textit{resp.}, decrease). In this way, the stochastic policy allows to explore the action space through random deviations from the mean, that are reinforced if they generate a profit.

\paragraph{The case of the risk insensitive agent}
For the risk insensitive agent, the gradient of the performance measure takes a simple form.

\begin{proposition}\label{prop:PolicyGrad}
For the risk insensitive agent, under a Dirichlet policy,
\begin{equation}
\nabla J\left(\mtrx{\theta}_{t}\right) =\sum_{n=1}^{N}\left(\mathbb{E}\left[r_{t+1,n}\mid\mtrx{X}_{t}\right]-\mathbb{E}_{\pi}\left[R_{t+1}\mid\mtrx{X}_{t},\mtrx{\theta}_{t}\right]\right)\frac{\nabla a_{t,n}}{\sigma_{t}}, \label{eq:sensi66}
\end{equation}
where we recall that $r_{t+1,n}$ is the return of asset $n$ between time $t$ and $t+1$ (while $R_{t+1}$ is the reward).
\end{proposition}
\begin{proof}
See Appendix \ref{sec:Proof-PolicyGrad}.
\end{proof}
We see that the learning process is ``\textit{myopic}'', as the one step ahead return is the only one taken into consideration for the update of the parameters. The learning process will, on average, increase (\textit{resp.}, decrease) the weights of those assets whose expected returns are higher (\textit{resp.}, lower) than the portfolio's expected return. This behavior could be expected from the risk-insensitive agent.

\subsection{The pricing ability of characteristics}
\label{sec:PAC}
 
 This section details additional theoretical properties of the Dirichlet policy. To ease notations and without much loss of generality, we assume that there is only one non-constant characteristic $\bm{c}_t$, which, at any point in time is distributed (across assets) symmetrically around zero. This is a common assumption in the asset pricing literature when agents are allowed to pre-process the data, see e.g., \cite{kelly2019characteristics}, \cite{gu2020empirical} and \cite{freyberger2020dissecting}.
 
The purpose of this subsection is to understand, in a simple case, the drivers of the variations of $\bm{\theta}_t$, which, via Equation \eqref{eq:pgrad}, is updated via the gradient times the learning rate. The strong assumption we make is that the agent implements the \textit{average} policy from Equation \eqref{eq:wgts} with parametric form (\textbf{F1}), so the vector of weights is $\bm{w}=\bm{a}/\sigma=\mtrx{X}\mtrx{\theta} /(\1^\intercal\mtrx{X}\mtrx{\theta} )$, where we omit the time index for notational clarity (for the remainder of the section). Essentially, this means that our results will hold \textit{on average}.

In all generality, the learning process will be based on several stages of subsampling, akin to bootstrapping (see Section \ref{sec:reinforce} below). Thus, we consider that the sum in Proposition \ref{prop:PolicyGrad} runs over a subset of the indices which we write $\mathcal{M} \subset \{1, \dots,N \}$ and which has cardinal $M$ such that $0\ll M\le N$. The reward, which we take to be the simple average return knowing $\bm{X}=[\1 \ \bm{c}]$, is
$$R\left(\theta^{(0)},\theta^{(1)}\right)=\E\left[ \frac{\sum_{m \in \mathcal{M}}r_m(\theta^{(0)}+\theta^{(1)}c_m)}{\sum_{m\in \mathcal{M}}(\theta^{(0)}+\theta^{(1)}c_m)}\right], \quad \text{card}(\mathcal{M})=M\le N,$$
so that the sensitivities from Equation \eqref{eq:sensi66} reduce to 
\begin{align}
\frac{\partial J}{\partial \theta^{(0)}} &=\sigma^{-1}\sum_{m \in \mathcal{M}}\left(\mathbb{E}\left[r_m\right]-R(\theta^{(0)},\theta^{(1)}) \right)=\sigma^{-1}\1^\intercal(\E[\bm{r}_\mathcal{M}]-R(\theta^{(0)},\theta^{(1)})\1), \label{eq:gradd1}\\
\frac{\partial J}{\partial \theta^{(1)}} &=\sigma^{-1}\sum_{m\in \mathcal{M}}\left(\mathbb{E}\left[r_m\right]-R(\theta^{(0)},\theta^{(1)}) \right)c_m=\sigma^{-1}\bm{c}^\intercal(\E[\bm{r}_{\mathcal{M}}]-R(\theta^{(0)},\theta^{(1)})\1), \label{eq:gradd2}
\end{align}
where $\bm{r}_\mathcal{M}$ is the return vector of assets belonging to the subset $\mathcal{M}$. Equation \eqref{eq:gradd1} implies that the parameter of the constant will decrease whenever the equally-weighted return is below that of the return of the policy. Equation \eqref{eq:gradd2} states that the parameter of the characteristic will vary with the relationship of the latter with the relative performance of the assets versus the average policy return. If the characteristic is linked with outperformance (\textit{resp.}, underperformance) with respect to the average policy return, the parameter will rise (\textit{resp.}, shrink). 

More directly, if the conditions of the Leibniz integral rule hold (which imposes mild integrability requirements on the returns), then the sensitivities of the reward are
  $$\frac{\partial }{\partial \theta^{(0)}}R(\theta^{(0)},\theta^{(1)})=\theta^{(1)} \Omega, \quad \frac{\partial }{\partial \theta^{(1)}}R(\theta^{(0)},\theta^{(1)})=-\theta^{(0)} \Omega,$$
 with 
 $$\Omega = \E\left[\frac{\sum_{m\in \mathcal{M}}r_m \times \sum_{m\in \mathcal{M}}c_m-M\sum_{m\in \mathcal{M}}c_mr_m}{\left( \sum_{m\in \mathcal{M}}(\theta^{(0)}+\theta^{(1)}c_m)\right)^2} \right]=M\E\left[\frac{\sum_{m\in \mathcal{M}}r_m \overbrace{\left(M^{-1} \sum_{l\in \mathcal{M}}c_l-c_m \right)}^{\text{long-short portfolio}}}{\left( \sum_{m\in \mathcal{M}}(\theta^{(0)}+\theta^{(1)}c_m)\right)^2} \right],$$
which allows for a similar interpretation. The derivatives of the reward linked to the policy will essentially be driven by the average return of a portfolio dictated by the relative values of the characteristics with respect to their mean.

We end our analysis by considering the case when all assets are included, i.e., $M=N$. The reward then simplifies to
 $$R(\theta^{(0)},\theta^{(1)})=\E\left[ \frac{(\theta^{(0)}\1 + \theta^{(1)}\bm{c})^\intercal\bm{r}}{(\theta^{(0)}\1+\theta^{(1)}\bm{c} )^\intercal\1}\right]=\E\left[ \frac{(\theta^{(0)}\1 + \theta^{(1)}\bm{c})^\intercal\bm{r}}{\theta^{(0)}N}\right]=(N^{-1}\1 + \theta^{(1)}\bm{c}/(\theta^{(0)}N))^\intercal\E\left[ \bm{r}\right],$$
 where in the second equality we have used the symmetry of $\bm{c}$ (the elements of which sum to zero). 
Again assuming that the conditions of the Leibniz integral rule are satisfied, this implies
 $$\frac{\partial }{\partial \theta^{(0)}}R(\theta^{(0)},\theta^{(1)})=-\theta^{(1)}\frac{\bm{c}^\intercal\E\left[\bm{r}\right]}{(\theta^{(0)})^2N} \quad \text{and} \quad \frac{\partial }{\partial \theta^{(1)}}R(\theta^{(0)},\theta^{(1)})=\frac{\bm{c}^\intercal\E\left[\bm{r}\right]}{\theta^{(0)}N},$$
 which means that the sensitivities of the reward are proportional to the characteristic-weighted average return, whereby we consider the values of $\bm{c}$ to be unconstrained long-short portfolio weights. In short, the sensitivities are driven by the way the characteristic is \textit{priced}.

Moreover, from Proposition \ref{prop:PolicyGrad}, we have that the derivative of the performance measure with respect to the policy parameters are
\begin{align}
\frac{\partial J}{\partial \theta^{(0)}} &=\sigma^{-1}\sum_{n=1}^{N}\left(\mathbb{E}\left[r_n\right]-(N^{-1}\1 + \theta^{(1)}\bm{c}/(\theta^{(0)}N))^\intercal\E\left[ \bm{r}\right] \right)=-\sigma^{-1}\theta^{(1)}(\theta^{(0)})^{-1}\bm{c}^\intercal\E[\bm{r}], \\
\frac{\partial J}{\partial \theta^{(1)}} &=\sigma^{-1}\sum_{n=1}^{N}\left(\mathbb{E}\left[r_n\right]-(N^{-1}\1 + \theta^{(1)}\bm{c}/(\theta^{(0)}N))^\intercal\E\left[ \bm{r}\right] \right)c_n=\sigma^{-1}\bm{c}^\intercal\E[\bm{r}],
\end{align}
where the very last equality comes from the symmetric distribution of $\bm{c}$ around zero. The overwhelming importance of the term $\bm{c}^\intercal\E[\bm{r}]$ incites us to coin a term for it.

\begin{definition}
\label{def:pac}
The pricing ability of a characteristic $\bm{c}$ is $\text{PAC}=\bm{c}^\intercal\E[\bm{r}]$.
\end{definition}
 
Note that if the characteristic is random, the $\text{PAC}=\E[\bm{c}^\intercal\bm{r}]$ is the covariance between the returns and the characteristics whenever the returns have zero mean. For a characteristic $k$ to matter in the portfolio construction process (assuming they are all normalized), it is required that the corresponding policy parameter $\theta^{(k)}$ be large in magnitude compared to other values of $\theta^{(j)}$. Since the process will be iterative through time, this requires that the gradient adjustments are such that $\partial J/\partial \theta^{(k)}$: \vspace{-2mm}
\begin{equation}
\left.\begin{array}{ll} 
1. &\text{ is larger in absolute value than } \partial J/\partial \theta^{(j)}, \ j\neq k; \\
2. &\text{ has a constant sign across consecutive dates}.
\end{array} \right\}
\label{eq:cond9}
\end{equation}
Even though these conditions originate from a stylized framework, our empirical results will confirm their practical relevance.


\section{Data and protocols}
\label{sec:2}

In this section, we describe the data on which we carry out our empirical analysis. Additionally, we discuss many implementation issues related to the RL framework. \vspace{1mm}

\footnotesize{\textbf{Data Availability Statement}. The data that support the findings of this study are available from Bloomberg LP. Restrictions apply to the availability of these data, which were used under license for this study. Data are however available from the authors upon reasonable request and with permission of Bloomberg LP.} \vspace{-2mm}

\normalsize


\subsection{Data}
\label{sec:data}
The dataset comprises firms listed in the US between January 2000 and June 2020 downloaded from Bloomberg. The number of firms through time is depicted in the left panel of Figure \ref{fig:nbfirms}. Observations are sampled at a monthly frequency. Average returns for each calendar year are shown in the left panel of Figure \ref{fig:nbfirms}.
Each stock is characterized by twelve attributes that correspond to documented predictors (accounting-based, risk-based and momentum-based). These variables are summarized in Table \ref{tab:pred}. We restrict our analysis to these twelve indicators to be able to easily comment on the associated values of $\bm{\theta}$. These features naturally serve as the non constant components of the $\bm{X}_t$ matrices in our model. 

\begin{figure}[!h]
\begin{center}
\includegraphics[width=15.5cm]{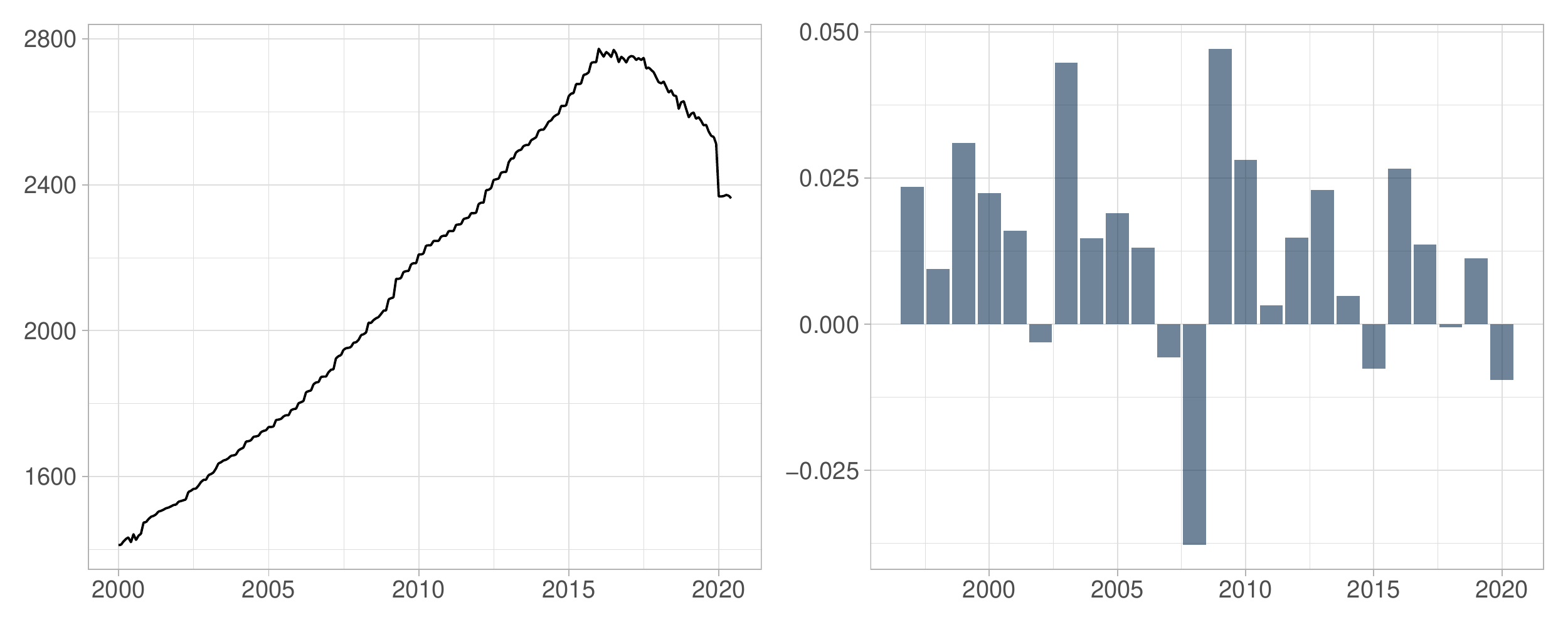}\vspace{-3mm}
\caption{\textbf{Number of firms and average equally-weighted returns}.} \vspace{-3mm}
\label{fig:nbfirms}
\end{center}
\end{figure}

The features (predictors) are cross-sectionally processed so that for a fixed $t$ and given predictor $j$, $\bm{x}_t^{(j)}$ is uniformly distributed (over the [-0.5,0.5] interval) across firms. Scaling predictors is standard practice both in the machine learning literature and in some recent asset pricing models (e.g., \cite{koijen2019demand}, \cite{kelly2019characteristics} and \cite{freyberger2020dissecting}). In characteristics-based investing (e.g., in \cite{brandt2009parametric}), the indicators are for instance also demeaned and standardized.

\begin{table}[!h]
\begin{center}
\scriptsize
\begin{tabular}{cll}
\textbf{Short name} & \textbf{Long name} & \textbf{Academic references} (alphab. order) \\ \hline
cst & Constant & -\\
cap &  Market Capitalization  & \cite{banz1981relationship, fama1992cross} \\
pb & Price-to-Book ratio & \cite{fama1992cross,asness2013value} \\
de & Debt-to-Equity ratio & \cite{bhandari1988debt,barbee1996sales} \\
vol & Realized volatility in the past 30 days & \cite{baker2011benchmarks} \\
prof & Profitability & \cite{fama2015five} \\
inv & Asset growth & \cite{cooper2008asset,fama2015five} \\
eps & Earnings per share & \cite{ball1968empirical, ball2019ball} \\
liq & Trading volume & \cite{chordia2000trading} \\ 
rsi & Relative strength index & \cite{han2013new} \\
pe & Price-earnings ratio & \cite{basu1983relationship,easton2004pe} \\
dy & Dividend yield & \cite{litzenberger1982effects}, \cite{naranjo1998stock} \\
mom & 12-1M momentum & \cite{jegadeesh1993returns, asness2013value} \\
 \hline
\end{tabular}
\caption{\textbf{List of predictors and associated academic references}. \small The Bloomberg fields are, in order, CUR\_MKT\_CAP, PX\_TO\_BOOK\_RATIO, TOT\_DEBT\_TO\_TOT\_EQY, VOLATILITY\_30D, PROF\_MARGIN, ASSET\_GROWTH, IS\_EPS, PX\_VOLUME, RSI\_30D, PE\_RATIO. The dividend yield is evaluated as EQY\_DPS divided by the lagged value of the closing price field PX\_LAST. Momentum is also computed via the closing price (lagged 12 month value divided by lagged one month value, minus one).\label{tab:pred}} \vspace{-3mm}
\end{center}
\end{table}

\subsection{The REINFORCE algorithm and implementation issues}
\label{sec:reinforce}
The learning process at the core of our method is the so-called REINFORCE algorithm, which is the most straightforward route towards the policy gradient approach (see Chapter 13 of \cite{sutton2018reinforcement}). We briefly recall the steps in Table \ref{tab:REINFORCE}. 

\begin{table}[!h]
\begin{center}
\begin{tabular}{clll}
Step &\multicolumn{3}{l}{\textbf{REINFORCE}} \\ \hline
0 &\multicolumn{3}{l}{Given a policy $\pi_{\bm{\theta}}$, a discount rate $\gamma\in(0,1)$ and  a learning rate $\eta\in(0,1)$;} \\
1 &\textbf{For} &\multicolumn{2}{l}{ i = 1, 2, $\dots$, number of episodes, do:} \\
2 && \multicolumn{2}{l}{Generate sequence $S_0$, $A_0$, $R_1$, $\dots$, $S_{T-1}$, $A_{T-1}$, $R_T$, with actions driven by $\pi_{\bm{\theta}}$} \\
3&& \textbf{For}&  t = 0, 1, $\dots,T-1$, do: \\
4&&&$G \leftarrow \sum_{k=t+1}^T\gamma^{k-t-1}R_k$ $\hspace{7mm}$ (compute the gain)  \\
5&&& $\bm{\theta} \leftarrow \bm{\theta} + \eta \gamma^t  G \nabla \log(\pi_{\bm{\theta}})$. $\quad$ (update the parameters)    \\ \hline
\end{tabular}
\caption{Steps of the baseline REINFORCE algorithm. \label{tab:REINFORCE}} \vspace{-6mm}
\end{center}
\end{table}

In spite of its apparent simplicity, the REINFORCE algorithm leaves a lot of room for implementation design. Below, we discuss open options (highlighted in bold font) for the steps defined in Table \ref{tab:REINFORCE}:
\begin{itemize}
\item \textbf{Step 0}: the policy is defined by Equation (\ref{eq:dirichlet}) along with one of the specifications (\textbf{F1}) or (\textbf{F2}). The levels of the two rates $\gamma$ and $\eta$ will be discussed below. 
\item \textbf{Step 1}: the central question is: what is an episode? More precisely, do we need to impose a chronological ordering of events? In traditional RL, this is imperative because actions can have an impact on the environment (the states). This is rarely the case in finance, except when taking very large orders, which large institutions usually avoid to limit the odds of market shifts. The generation of the SARSA sequences in Step 2 can thus be either \textbf{chronological} indeed, or independently drawn from samples of features and returns, akin to \textbf{bootstrapping}. In the latter case, the discounting rate $\gamma$ would lose its meaning and should be set at one.
\item \textbf{Step 2 \& 4}: the definition of the reward $R_t$ is not unambiguous. A natural choice is to take raw \textbf{returns}. The most prominent extension is when returns need to be adjusted by some risk measure, like in the \textbf{Sharpe ratio} (SR). However, the computation of the rewards in this case is not straightforward. Luckily, \cite{moody1998performance} provide a solution to this obstacle. Their idea is to sequentially update the reward using the exponential moving average SR given in Equation~\eqref{eq:SR}.

With this convention, the first steps of the SARSA sequence rapidly yield a risk adjusted reward. 
\item \textbf{Step 5}: this is not an option, but in the case of the linear policy $(\textbf{F1})$, the updated $\bm{\theta}$ has to be adjusted so that the $a_n$ lie in the intervals discussed in Section \ref{sec:a} in the Appendix. Since we will work with bundles of $N=100$ assets, we choose $a_-=0.02$ and $a_+=1.6$. Thus, the feasible set in the projection (\ref{eq:proj}) is  \vspace{-2mm}
$$\Theta(\bm{X}_t)=\{\bm{\theta} , \ a_-\le  \bm{\theta X}_t \le a_+ \}.$$
\end{itemize}

\subsection{Protocols}
\label{eq:prot}
The above discussion gives rise to two dichotomies: \textbf{chronological} versus \textbf{bootstrapped} sequences and \textbf{return} versus \textbf{Sharpe ratio} rewards. Below, we explain how to incorporate these design choices in a series of thorough backtests that rely on market data (and not on simulated samples).

Chronological sequences require temporal depth. Every January (time $t$), the preceding 12 months of data are gathered and the sequences will consist of portfolios held during each month in the sample. The number of episodes is $E$ and their length is 12. Each action at month $s$ ($A_s$) consists in randomly choosing (with replacement) $N$ assets and sampling their portfolio weights according to the current policy $\pi_{\bm{\theta}}$. The weights depend both on $\bm{\theta}$ and on the characteristics of the assets at month $s$. Rewards can either be returns, or Sharpe ratios, computed iteratively as defined in Equation (\ref{eq:SR}).

Bootstrapped sequences do not require much depth. They can be performed every month. The number of episodes is $E$ and their length is one: the learning is performed over values that originate from the past month only. This could be relaxed, but it creates more reactive portfolios, as opposed to learning on chronological sequences. Again, actions $A_s$ consist in randomly choosing (with replacement) $N$ assets and sampling their portfolio weights according to the current policy $\pi_{\bm{\theta}}$. Rewards can only be returns. 

Both types of learning processes are summarized in Table \ref{tab:backtest}. Because of the numerous degrees of freedom ($\gamma$ and $\eta$ rates, initialization values, random seeds, number of episodes, etc. - see Section \ref{sec:deg} below), we restrict our study to two alternatives only. The first one links bootstrapped sequences with simple returns, while the second combines chronological sequences with Sharpe ratio rewards.

\begin{table}[!h]
\small
\begin{minipage}[b]{0.5\hsize}
\begin{tabular}{clllll}
Step &\multicolumn{4}{l}{\textbf{Chronological method}} \\ \hline
0 &\multicolumn{4}{l}{\textbf{For} every January, do:} \\
1 &  &\multicolumn{4}{l}{Extract data from previous year } \\
2 &  &\multicolumn{4}{l}{Randomly pick $N$ assets } \\
3& &\multicolumn{4}{l}{Initialize $\bm{\theta}$ } \\
4& &\multicolumn{4}{l}{\textbf{For} i $=1,\dots$ episodes, do} \\
5 &&& \multicolumn{3}{l}{Sample $N$ stocks randomly} \\
6 &&& \multicolumn{3}{l}{Generate streams $A_t$ and $R_t$ via $\pi_{\bm{\theta}}$} \\
7 &&& \multicolumn{3}{l}{Update $\bm{\theta}$ via \eqref{eq:pgrad}} \\
8&& \textbf{For}&  next 12 months, do: \\
9&&&allocate via average policy, Eq. (\ref{eq:wgts})\\
10&&& store realized returns    \\ \hline
\end{tabular}
\end{minipage}
\begin{minipage}[b]{0.5\hsize}
\begin{tabular}{clllll}
Step &\multicolumn{4}{l}{\textbf{Bootstrap method}} \\ \hline
0 &\multicolumn{4}{l}{\textbf{For} every date $t=2,\dots,T-1$, do:} \\
1 &  &\multicolumn{4}{l}{Extract data from previous month } \\
2 &  &\multicolumn{4}{l}{Randomly pick $N$ assets} \\
3& &\multicolumn{4}{l}{Initialize $\bm{\theta}$ } \\
4& &\multicolumn{4}{l}{\textbf{For} i $=1,\dots$ episodes, do} \\
5 &&& \multicolumn{3}{l}{Sample $N$ stocks randomly} \\
6 &&& \multicolumn{3}{l}{Generate action and reward } \\
7 &&& \multicolumn{3}{l}{Update $\bm{\theta}$ via \eqref{eq:pgrad}}  \\
8&& \textbf{For}&  date $t+1$, do: \\
9&&&allocate via average policy, Eq. (\ref{eq:wgts})\\
10&&& store realized returns    \\ \hline
\end{tabular}
\end{minipage}
\caption{Macro view of backtest stages. The differences in the two REINFORCE implementations are outlined in Section \ref{eq:prot}. \label{tab:backtest}}
\end{table}


\section{Results}
\label{sec:3}

\subsection{Degrees of freedom}
\label{sec:deg}

Before we move towards a presentation of our results, we expose the richness of the flexibility of the modelling approach. Below, we list the different choices we need to make to run one batch of learning over our whole dataset:
\begin{itemize}
\setlength\itemsep{-0.3em}
\item Choices that we will always compare: \vspace{-3mm}
\begin{enumerate}
\setlength\itemsep{-0.3em}
\item Whether to learn form chronological sequences or bootstrap (see Table \ref{tab:backtest}).
\item Whether to resort to a linear (\textbf{F1}) or an exponential (\textbf{F2}) policy. 
\end{enumerate}
\item Choices that we will discuss: \vspace{-3mm}
\begin{enumerate}
\setlength\itemsep{-0.3em}
\item The number of episodes.
\item The initialization values of $\bm{\theta}_0$.
\item The seed for the quasi-random number generator.
\item $\eta$, the learning rate in the update of the policy parameters. To simplify scale issues, the gradient in the update is divided by the maximum absolute value of  gradient values. This makes the learning rate easier to interpret. 
\end{enumerate}
\item Choices that are fixed throughout the entire study: \vspace{-3mm}
\begin{enumerate}
\setlength\itemsep{-0.3em}
\item $\gamma$, the reward discounting factor. For bootstrap learning, this parameter is irrelevant. For chronological sequences, since they only last 12 months, there is no major nor obvious gain in using a discount. Thus we set $\gamma=1$.
\item $N$, the number of stocks that are integrated in the portfolio (used to compute the reward). As discussed in Section \ref{sec:a}, it is impossible to consider very large portfolios because of the asymptotic behavior of the functions required in the Dirichlet forms. The most obvious choice is $N=100$. Larger portfolios impose stringent constraints on the Dirichlet parameters, making the approach impractical. By construction, smaller portfolios lack diversification and may reflect cross-sectional information insufficiently.
\item The bounds on the Dirichlet parameters (see Section \ref{sec:a} in the Appendix). They are fixed to $a_-=0.2$ and $a_+=1.6$. These bounds are optimal empirically: going beyond leads to numerical errors. 
\item \textit{Rewards}. Bootstrapped sequences can only work with simple return rewards. Chronological sequences are more flexible. To reduce the amount of results, we work with the differential Sharpe ratio for temporal learning.
\end{enumerate}
\end{itemize}

\subsection{Baseline output: factor coefficients and Dirichlet parameters}

First and foremost, the Dirichlet policy depends on its parameter vector $\bm{\theta}$. It is thus natural to start by showing the evolution of the $\theta_t^{(k)}$ for the four specifications we work with. They are shown in Figure \ref{fig:theta}. While only a few cases are outlined, they are qualitatively representative of all the other parameter configurations studied below. 

\begin{figure}[!h]
\begin{center}
\includegraphics[width=15cm]{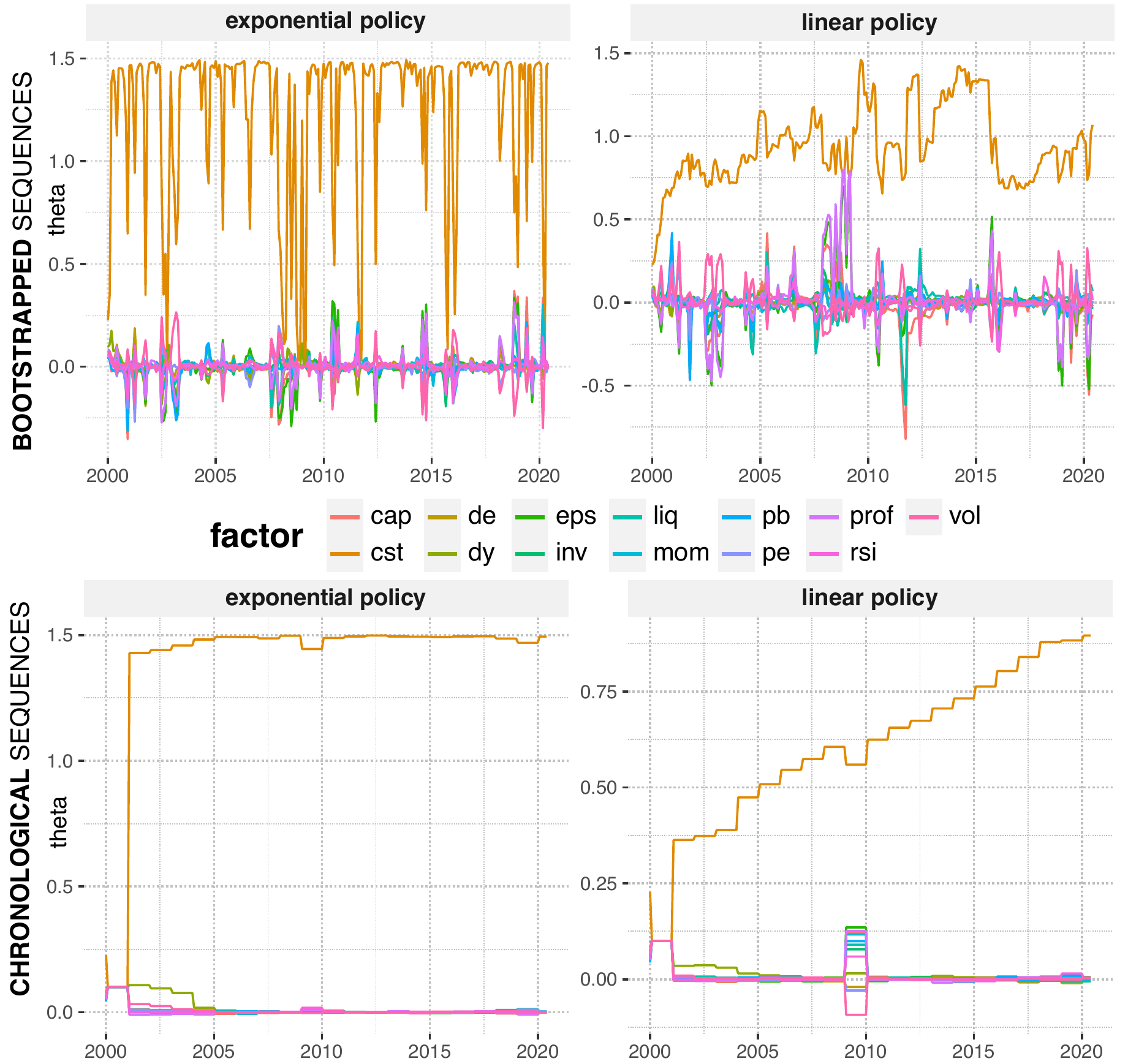}\vspace{-3mm}
\caption{\textbf{Values of $\theta^{(k)}_t$}. \small We plot the value of parameters through time for our two learning schemes (bootstrapped (upper panel) and chronological (middle panel)) and two policy schemes (linear (right) versus exponential (left) - see Equation (\ref{eq:policies})). The parameters are the following: the learning rate $\eta = 0.1$, the number of episodes $E=500$, the bounds for the $a_n$ are $[0.2,1.6]$, the initial value for all $\theta^{k}$ is 1. Finally, the random seed in 42.}
\label{fig:theta}
\end{center}
\end{figure}

There is a clear discrepancy between the two learning schemes: chronological sequences (lower panels) lead to the hegemony of the constant variable while the bootstrapped sequences (upper plots) give more room to the firm characteristics. The latter are also much more volatile through time. Across both learning methods, the exponential policy (to the left) saturates the constant much more often, compared to the linear policy (to the right).

This has consequences on the optimal weights derived from the policy parameters via Equation (\ref{eq:policies}). The chosen portfolio weights are simply chosen as the mean of the policy distributions: $w_{t,n}=a_{t,n}/\sigma_t$ (see equation \eqref{eq:wgts}). In Figure \ref{fig:wgts}, we plot the histogram of these weights. The distributions are grouped by year and then stacked on the graph. Because the number of assets changes through time (see Figure \ref{fig:nbfirms}), we add two bounds on the plots. The full vertical black line marks the minimum uniform allocation (1/$N$), which is reached in 2016. The dotted line shows the maximum 1/$N$ weights, which are implemented in 2001.

\begin{figure}[!h]
\begin{center}
\includegraphics[width=15cm]{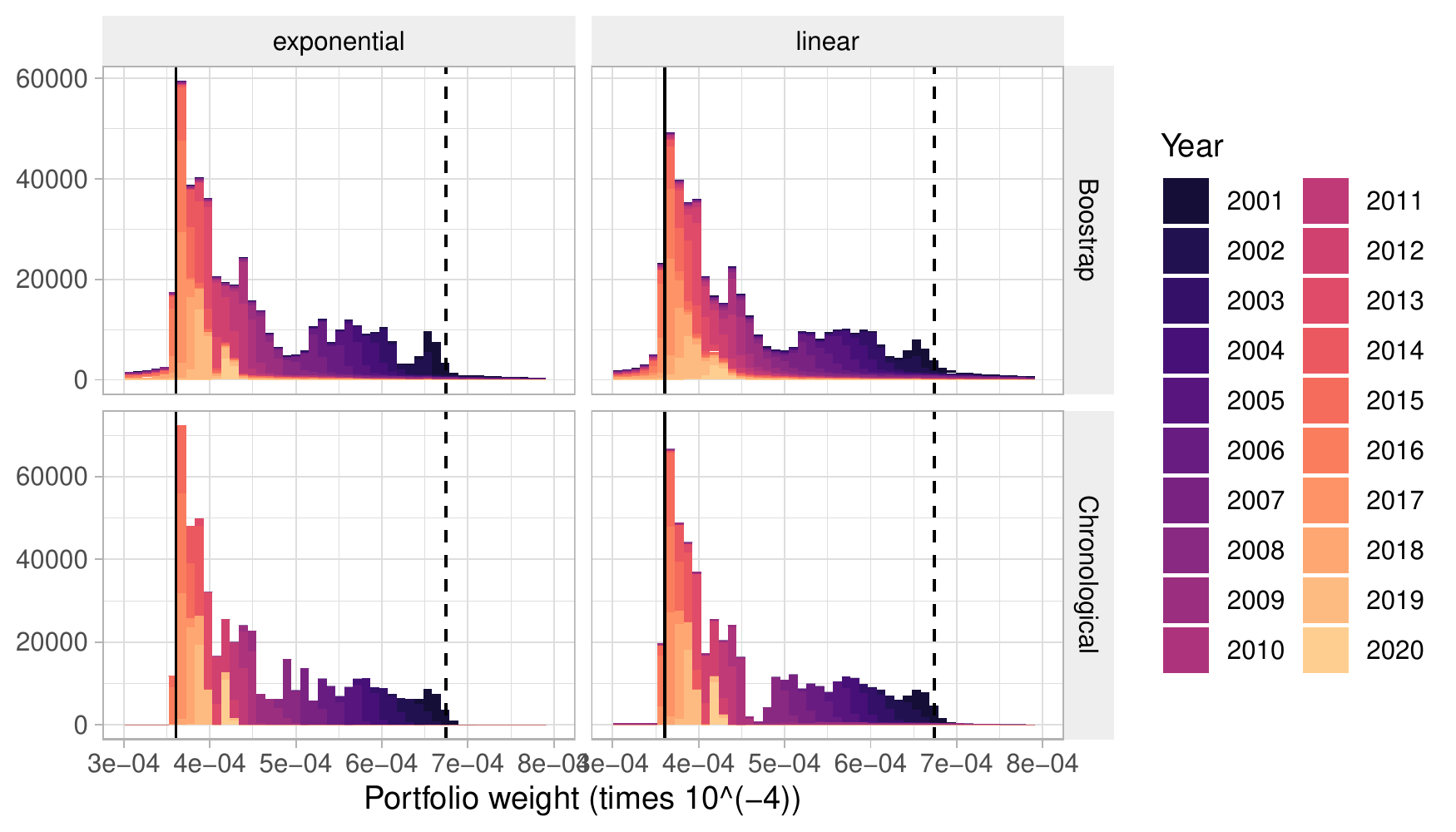}\vspace{-3mm}
\caption{\textbf{Distribution of weights}. \small We plot the histogram of portfolio weights for our two learning schemes (bootstrapped (upper panel) and chronological (middle panel)) and two policy schemes (linear (right) versus exponential (left) - see Equation (\ref{eq:policies})). The histograms are stacked and each color stands for a given year. The full vertical line marks the minimum uniform allocation (1/$N$) over all dates, while the dotted line shows the maximum 1/$N$ value. The parameters are the following: the learning rate $\eta = 0.1$, the number of episodes $E=500$, the bounds for the $a_n$ are $[0.2,1.6]$, the initial value for all $\theta^{k}$ is 1. Finally, the random seed in 42.} 
\label{fig:wgts}
\end{center}
\end{figure}

First of all, because of the increase in the number of stocks, there is a temporal shift in the distribution of weights. Average weights are smaller in the later years and portfolios are more diversified. Moreover, weights are not very dispersed and appear concentrated around their means, which implies that allocations are relatively close to the EW benchmark and do not make strong bets towards some assets. This is especially true for the lower panel (chronological sequences), where there are almost no outliers beyond the vertical lines. This is consistent with the prominence of the constant in the lower panels of Figure \ref{fig:theta}. 

Again, we underline that these results depend only marginally on the parametric choices described in the caption of the figures. The concentration of portfolios does not depend much implementation choices, as long as they are realistic (e.g., sufficiently many episodes, or moderate learning rate).\footnote{Additional results are available upon request. }

\subsection{Portfolio performance}

The ultimate yardstick for sophisticated portfolio construction methods is out-of-sample performance. It is usually presented in several steps: starting with a pure return indicator, and complementing it by other risk-adjusted metrics, like the Sharpe ratio.
In Figure \ref{fig:rets}, we display average realized returns (left panels) and Sharpe ratios (right plots) of the mean policy for couples of values for random seeds, learning rate, and parameter initialization.

\begin{figure}[!h]
\begin{center}
\includegraphics[width=16.5cm]{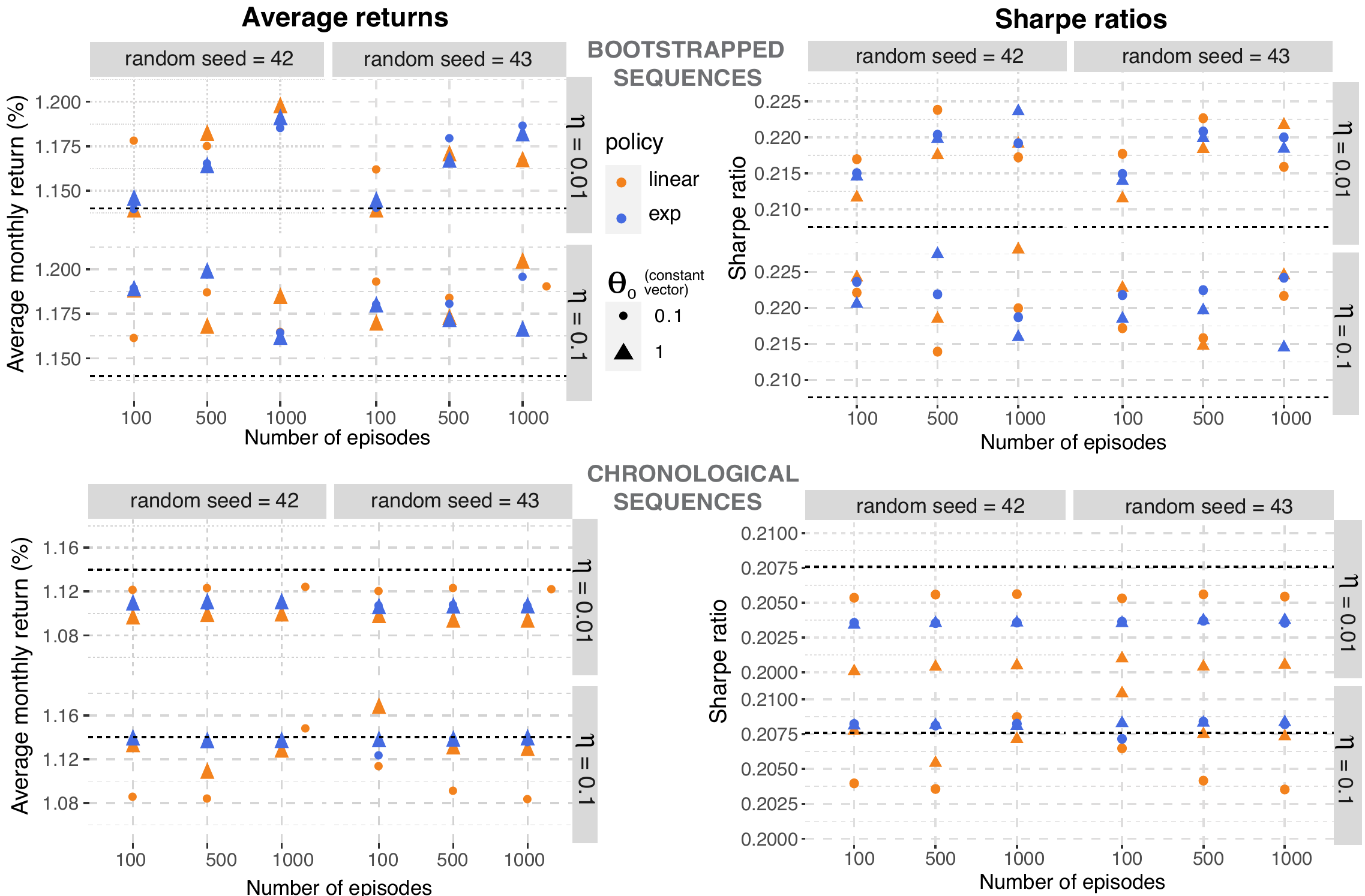}\vspace{-3mm}
\caption{\textbf{Performance}. \small We plot the average returns (left quadrants) and Sharpe ratios (right quadrants) of the portfolio allocations for our two learning schemes (bootstrapped (upper panel) and chronological (lower panel)), and two policy schemes (linear (orange points) versus exponential (blue points) - see Equation (\ref{eq:policies})). The dotted horizontal line marks the performance of the EW (1/$N$) portfolio.} 
\label{fig:rets}
\end{center}
\end{figure}

The only robust conclusion is that bootstrap sequences perform better than sequential ones. Almost all bootstrap configurations surpass the EW benchmark, while it is the opposite for the portfolios based on chronological learning. One reason for this may be that bootstrap learning is more reactive, while sequential learning will learn from older data. The four degrees of freedom (random seed, number of episodes, learning rate and $\bm{\theta}_0$) have an impact on average returns that is not consistent across configurations. Notably, this corroborates the sensitivity of RL algorithms to random seeds (see, e.g., \cite{henderson2017deep}, \cite{islam2017reproducibility}, and \cite{colas2018many}). Nevertheless, the magnitude of changes is small overall: average returns are scattered between 1.08\% and 1.2\%, so that the difference with the uniform allocation (1.14\%) is not significant.\footnote{A simple $t$-test of series of RL-based portfolio returns versus $1/N$ portfolio returns yields $p-$values between 0.82 and 0.998.} Thus, even though parameter configurations alter the results, the changes are very limited in magnitude and all portfolios remain somewhat in the vicinity of the $1/N$ benchmark.

We end this subsection with some comments on transaction costs, which can be a major issue when portfolio turnover soars. Fortunately, with policies which remain in the vicinity of the equally-weighted portfolio, this is not the case. In all our backtests, turnover is evaluated as the average monthly asset rotation:
$$\text{Turn}=\frac{1}{T}\sum_{t=1}^TN_{t}^{-1}\sum_{n=1}^{N_t}|w_{t,n}-w_{t-,n}|,$$
where $w_{t-,n}$ is the weight of asset $n$ just before rebalancing. In terms of magnitude, the Dirichlet policies have turnovers around 9\%-10\%, while capitalization-weighted portfolios oscillate around 20\%, on a monthly basis. This may seem surprising, as cap-weighted allocations are known to be extremely efficient with respect to trading costs. This is however true only when the investment universe is fixed. When the set of assets changes, cap-weighted portfolios are more penalized (see, e.g., Table 5 of \cite{demiguel2009generalized}). The cost of the trading, if, as in \cite{goto2015improving}, we conservatively assume a 50 basis point fee per unit of turnover, will be 5 basis points per month for the Dirichlet policies, which is very reasonable, compared to the roughly 110 basis points of monthly returns that are generated on average.

In unreported results,\footnote{These results are available upon request.} we built a new learning scheme aimed to penalize asset rotation. This can be done via a correction of the reward $R_t$ that subtracts the transaction costs incurred by the update in portfolio weights. However, the results were disappointing, possibly for two reasons. First, because the changing set of assets is likely to blur the learning process. Second, because the characteristics are not able to capture the drivers of transaction costs.



\section{Discussion}
\label{sec:4}

The main empirical conclusion from the above exercise is underwhelming because a sophisticated machinery produces a simplistic outcome. This resonates with earlier contributions which also document that RL is able to ``\textit{rediscover}'' known results.\footnote{In more complex situations, \cite{chaouki2020deep} and \cite{kong2018new} have shown that RL is able to solve mainstream optimization problems.} From a pure performance standpoint, equally-weighted portfolios have been documented to be solid benchmarks in the long run (see \cite{demiguel2009optimal}, and \cite{plyakha2015equal}). Therefore, our RL strategy, while incapable of timing the factors, nonetheless suggests a sensible allocation.

\subsection{Characteristics are weakly priced}

One reason why the constant characteristic dominates in RL portfolio may simply be \textit{noise}. During each episode, the bootstrap procedure selects stocks randomly and the algorithm extracts the gradient that works best for these stocks. Unfortunately, in the subsequent episode, the relationship between returns and predictors may very well be completely different, which will attenuate the effect because the new gradient is likely to cancel the previous one. This is linked to the absence of arbitrage. If one variable (e.g., price-to-book) consistently predicted the cross-section of returns, then it would be easy to generate certain profits.

Empirically, the only dominating effect on markets is the equity premium, according to which returns in excess of the risk-free rate are on average positive (with long term means between 3\% and 10\% annually, depending on studies and markets).\footnote{To a lesser extent, \cite{smith2021have} show that momentum is the only ``risk factor'' which is associated to persisting average returns. The size and value premia seem to have completely disappeared. This may be linked to alpha decay and we refer to \cite{chordia2014have}, \cite{mclean2016does}, \cite{jacobs2020anomalies}, \cite{penasse2020understanding} and \cite{shanaev2021efficient} for more details on this matter.} Raw returns are higher and \cite{ilmanen2011expected} reports that the arithmetic return on the US equity market was above 12\% over the whole XX$^{th}$ century.

This explains why in Figure \ref{fig:theta}, the curves for the constant shrink after financial crises. In the upper panels, the monthly samples are more reactive and it is clear that in 2008 (left plot) the strong negative returns penalize the constant term. The declines of $\theta^{(\text{cst})}$ in 2015 and 2018 also correspond to years of negative returns for the US equity market (see Figure \ref{fig:nbfirms}). Thus, it is only in bad periods that the other characteristics have enough room to impact the allocation scheme.

This effect is also linked to the pricing ability of characteristics (PAC) introduced in Section \ref{sec:PAC}. In Figure \ref{fig:PAC} below, we plot a proxy for the \textit{realized} PAC, which we define as $\text{PAC}_t^{(k)}=N_t^{-1}\sum_{n=1}^{N_t}r_{t+1,n}x_{t,n}^{(k)}$, where $N_t$ is the number of asset in the dataset at time $t$.

\begin{figure}[!h]
\begin{center}
\includegraphics[width=15cm]{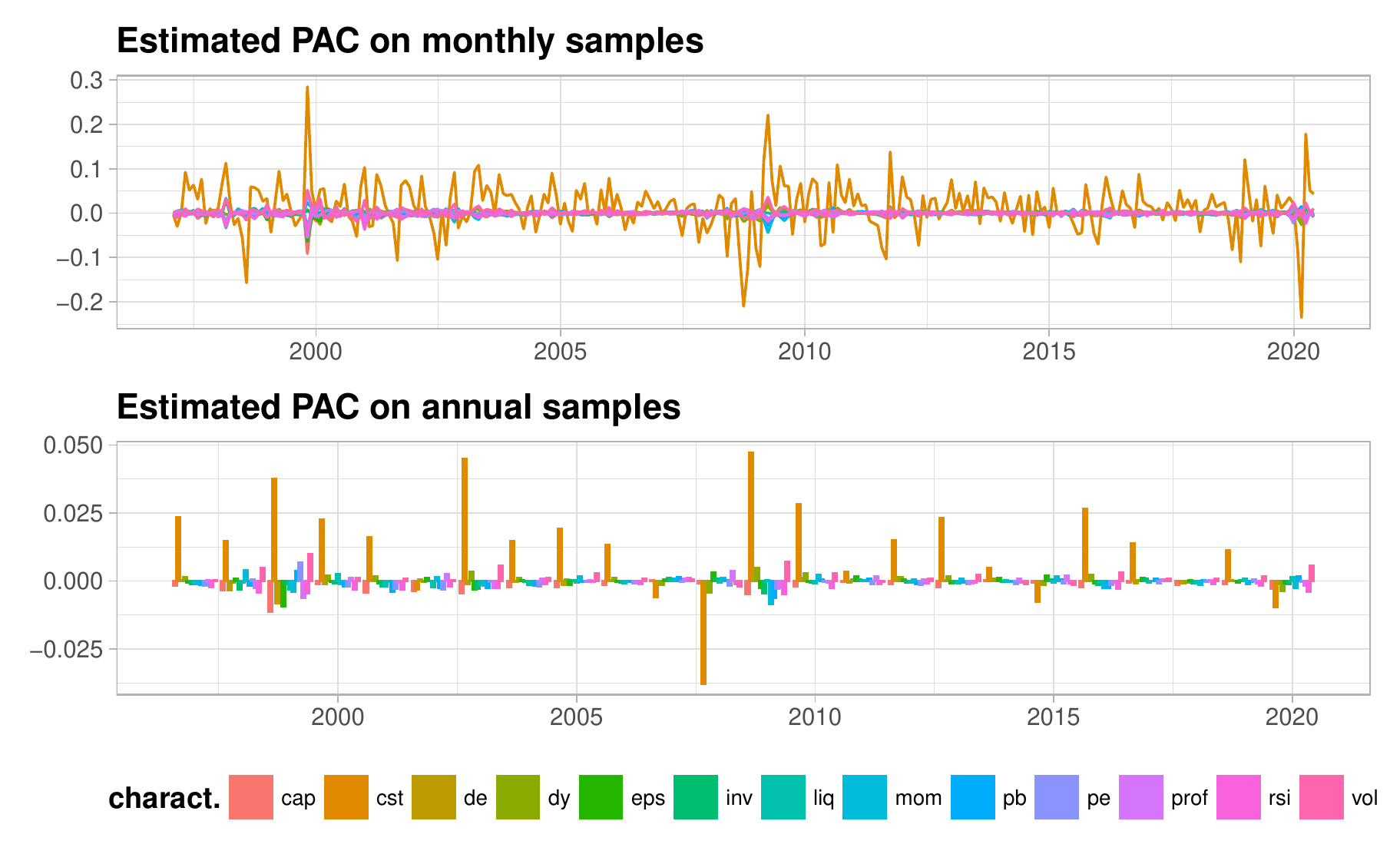}\vspace{-3mm}
\caption{\textbf{Pricing ability of characteristics}. \small We plot the realized PAC on monthly samples in the upper panel. For annual samples, in the lower panel, the formula is $\text{PAC}_t^{(k)}=\frac{1}{12}\sum_{s=0}^{11}N_{t-s}^{-1}\sum_{n=1}^{N_{t-s}}r_{t+1-s,n}x_{t-s,n}^{(k)}$.} 
\label{fig:PAC}
\end{center}
\end{figure}

The first qualitative finding is that the magnitude of the PAC is much larger for the constant. This is somewhat logical, because the PAC in this case relates to a long-only portfolio, while for the other characteristics, the PAC is the average return of a long-short allocation. Nevertheless, if one characteristic was strongly priced (i.e., if it comoved substantially with returns in the cross-section of assets), it would be associated with large PAC values, at least for a short period of time (as expressed in Equation \eqref{eq:cond9}). In addition, as expected from the equity premium, the PAC values are often positive for the constant characteristic. This is particularly salient in the lower panel, which is why its parameter increases in Figure \ref{fig:theta}. Because the chronological learning process is only updated every January, the pronounced market loss of 2008 is only reverberated in 2009, where the parameter for the constant sustains a small plunge.

To conclude this section, we plot in Figure \ref{fig:theta1} the evolution of $\theta^{(k)}_t$ when the characteristics are uniformly distributed on the \textit{unit} interval. In this case, the PACs correspond to long-only portfolios where the weights are proportional to the characteristics' values. Again, the prominence of the constant term is manifest, which is yet another proof that the characteristics used in the traditional asset pricing literature are relatively weakly priced. 

\begin{figure}[!h]
\begin{center}
\includegraphics[width=15cm]{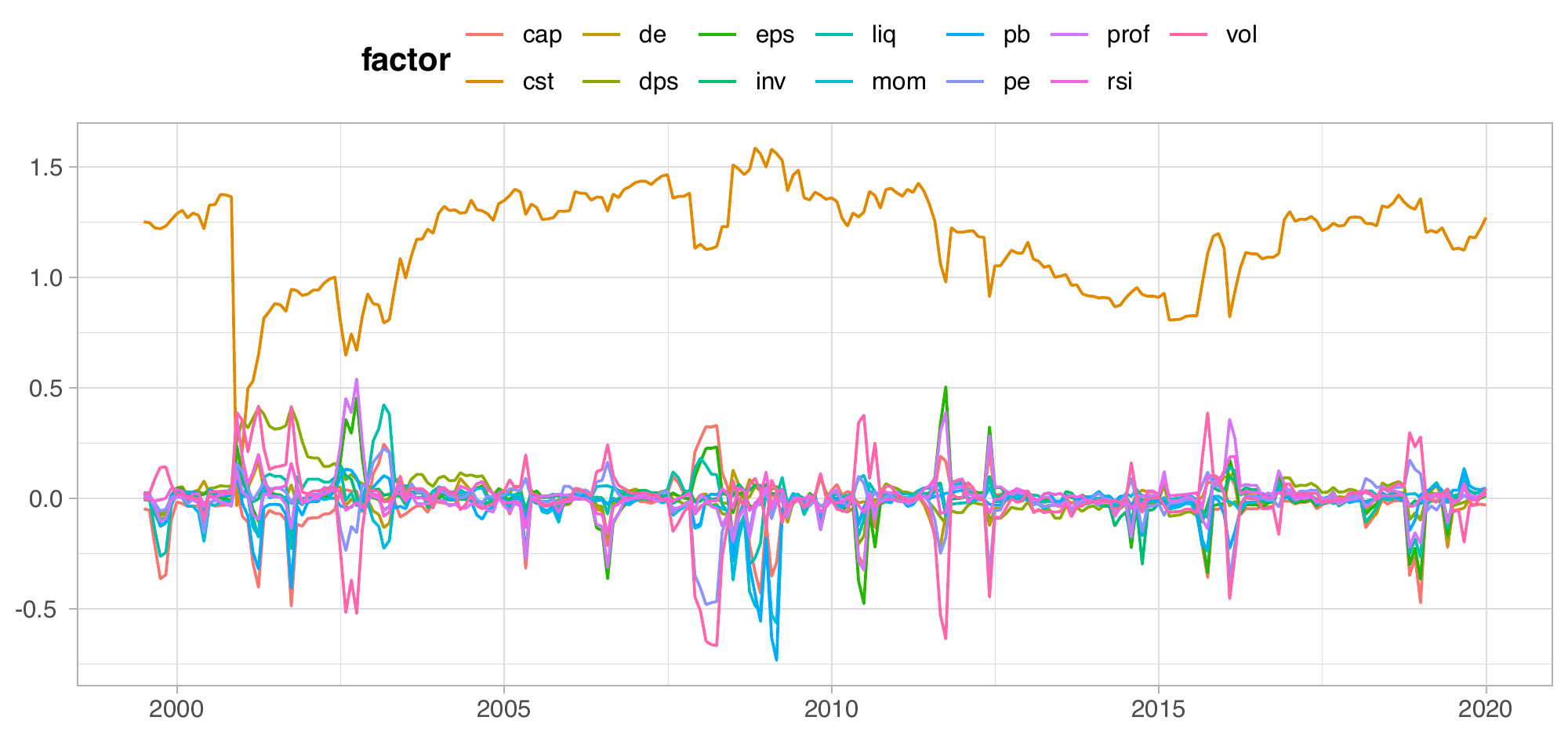}\vspace{-3mm}
\caption{\textbf{Values of $\theta^{(k)}_t$ when characteristics lie in [0,1]}. \small We plot the value of parameters through time for the linear bootstrap learning scheme. The parameters are the following: the learning rate $\eta = 0.1$, the number of episodes $E=500$, the bounds for the $a_n$ are $[0.2,1.6]$, the initial value for all $\theta^{k}$ is 1. Finally, the random seed in 42.}
\label{fig:theta1}
\end{center}
\end{figure}


\subsection{Insights from a toy factor model}

The purpose of this section is to further elaborate on the patterns observed in Figure \ref{fig:theta}, which shows that the most important driver of the policies is the constant term. This preponderance implies that RL-based allocations remain in close to the equally-weighted portfolio. We propose a factor-based allocation model that tries to replicate this stylized property. In particular, we investigate if RL-driven decisions can be reproduced by simpler models. We start with a theoretical contribution and subsequently move towards simple statistical estimates to illustrate the forces at work.

We assume that there are $N$ assets on the market. Their future returns are driven by a linear model
\begin{equation}
\bm{r}=\bm{X \beta} + \bm{\epsilon},
\label{eq:model}
\end{equation}
where $\bm{X}=X_{nk}$ is a $N \times(K+1)$ matrix of firm-specific characteristics with $N>K+1$. We omit the time index for notational simplicity. The first column of the matrix is constant with all elements equal to one. The innovations $\bm{\epsilon}$ and loadings $\bm{\beta}=(\beta_0,\beta_1,\dots,\beta_K)$ are random and mutually independent. Moreover, we posit that the errors are independent across assets (and independent of loadings), and have zero means and uniform variance: 
\begin{align}
\bar{\bm{\epsilon}}&=\E[\bm{\epsilon}]=\bm{0}_N \label{eq:mean} \\
\bm{\Sigma}_\epsilon&=\E[\bm{\epsilon \epsilon}^\intercal]=\sigma_\epsilon^2\bm{I}_N, \label{eq:variance}
\end{align}
where $\bm{0}_N$ is a $N$-dimensional vector of zeroes and $\bm{I}_N$ the corresponding identity matrix. For analytical tractability concerns, we also need to be more specific with regard to the covariance structure of loadings and characteristics. We assume that both
\begin{align}
\bm{\Sigma}_\beta&:=\E[( \bm{\beta}-\bar{\bm{\beta}})(\bm{\beta}-\bar{\bm{\beta}})^\intercal] \quad \text{and} \\
\hat{\bm{\Sigma}}_X&:= N^{-1}(\bm{X}-\bm{1}_N\bar{\bm{x}}^\intercal)^\intercal(\bm{X}-\bm{1}_N\bar{\bm{x}}^\intercal) \label{eq:XX}
\end{align}
are diagonal with diagonal values equal to $\sigma_{\beta,k}^2$ and $\sigma_{X,k}^2$ respectively, where we have casually written $\bar{\bm{\beta}}$ for the mean vector of $\bm{\beta}$ and $\bar{\bm{x}}=N^{-1}\bm{X}^\intercal\bm{1}_N$ for the column vector of sample column means of $\bm{X}$. The function diag$(\cdot)$ maps a vector into the corresponding diagonal matrix. Note that since $\bm{X}$ is given and non-random, the matrix $\hat{\bm{\Sigma}}_X$ is its \textit{sample} covariance matrix. The fact that $\hat{\bm{\Sigma}}_X$ is diagonal implies that the firm characteristics are uncorrelated, i.e., that they carry information pertaining to companies that is not redundant. The $\beta_k$ are also unrelated, which means that each factor impacts returns regardless of the effect of other firm attributes. 

Technically, the model is linked to the PAC defined in Definition \eqref{def:pac}. Indeed, the OLS estimator for $\bm{\beta}$ reads $\hat{\bm{\beta}}=(\bm{X}^\intercal\bm{X})^{-1}\bm{X}^\intercal \bm{r}$, where the second component $\bm{X}^\intercal \bm{r}$ is, up to a scaling factor, the vector of estimated PACs (the first element being the sum of returns). This implies that the estimated loadings are linear combinations of the pricing abilities of characteristics. Notably, if the characteristics are centered and independent, then $(\bm{X}^\intercal\bm{X})^{-1}$ is diagonal and the loading for a given characteristic is simply proportional to its PAC.

In this framework, some representative agent seeks to maximize a standard quadratic function of expected returns. The portfolio is based on firms' characteristics in a linear fashion: $\bm{w}=\bm{X\theta}$, where $\bm{\theta}=(\theta^{(0)},\theta^{(1)},\dots,\theta^{(K)})$ drives and reflects the agents beliefs and preferences with regard to the corresponding factors. This form is on purpose the same as \textbf{F1} in Equation \eqref{eq:policies} (Section \ref{sec:diripol}), which drives the average portfolio allocation. The utility function is quadratic (as in the standard mean-variance formulation), hence, the optimization program is the following:
\begin{equation}
\underset{\bm{\theta}}{\max} \ \E\left[ \bm{\theta}^\intercal\bm{X}^\intercal\bm{X}\bm{\beta}-\frac{\gamma}{2}\bm{\theta}^\intercal\bm{X}^\intercal( \bm{X} (\bm{\beta}-\bar{\bm{\beta}})+\bm{\epsilon})(\bm{\epsilon}^\intercal+(\bm{\beta}-\bar{\bm{\beta}})^\intercal \bm{X}^\intercal)\bm{X}\bm{\theta} \right], \quad \text{s.t.} \quad \bm{\theta}^\intercal\bm{X}^\intercal\bm{1}_N=1.
\label{eq:maxprog} 
\end{equation}

The lemma below provides the solution to this problem.

\begin{lemma}
\label{lem:ass}
If $\bm{\Sigma}_\beta$, $\bm{\Sigma}_\epsilon$ are diagonal and assuming \eqref{eq:mean}-\eqref{eq:variance}, the solution to \eqref{eq:maxprog} is 
\begin{equation} \normalfont
\tilde{\bm{\theta}}_*=\gamma^{-1}\text{diag}(\bm{\sigma}^2)^{-1}\left(\bm{I}_K- \frac{\text{diag}(\bm{\sigma}^2_\beta) \bar{\bm{x}}\bar{\bm{x}}^\intercal \text{diag}(\bm{\sigma}^2)^{-1}}{1+\bar{\bm{x}}^\intercal \text{diag}(\bm{\sigma}^2_\beta)\text{diag}(\bm{\sigma}^2)^{-1}\bar{\bm{x}}}\right) \left(N^{-1}\bar{\bm{\beta}}+  c(\bm{X}^\intercal \bm{X})^{-1} \bar{\bm{x}} \right), \label{eq:thetastar0}
\end{equation}
where $ \normalfont\text{diag}(\bm{\sigma}^2)=\text{diag}(\bm{\sigma}^2_X)\text{diag}(\bm{\sigma}^2_\beta)+N^{-1}\sigma_\epsilon^2 \bm{I}_K$ and $c$ is the scaling constant that warrants the budget constraint is satisfied. If, in addition, $\bar{\bm{x}}^\intercal=[1 \quad \bm{0}_K^\intercal]$, then
\begin{equation}
\tilde{\bm{\theta}}_*=\left\{ \begin{array}{ll}
\tilde{\theta}^{(0)}_* &=N^{-1}\vspace{2mm}\\ 
\tilde{\theta}^{(j)}_* &= (\gamma N)^{-1}  \frac{\bar{\beta}_j}{\sigma_{X,j}^2\sigma_{\beta,j}^2+\sigma_\epsilon^2/N}
\end{array} \right. .  
\label{eq:thetastar}
\end{equation}
\end{lemma}

The proof of the lemma is located in Appendix \ref{lem:proof}. All other things equal, in the second part of the lemma, $\tilde{\theta}^{(j)}_* $ increases with $\bar{\beta}_j$, but decreases with all sources of risk: $\sigma_{X,j}^2$, $\sigma_{\beta,j}^2$, and $\sigma_\epsilon^2$. More importantly, the relative importance of the non-constant factors are strongly linked to $\gamma$. When risk aversion is low, the non-constant factors play a prominent role in the allocation choice. If, however, risk aversion is high, then the $\tilde{\theta}_*^{(j)}$ are negligible and the $1/N$ portfolio is appealing to the investor. Based on our empirical results, the latter situation seems more likely. 

One particular subcase of the lemma is when the budget constraint (to the right of Equation \eqref{eq:maxprog}) is removed. In the general case, this implies $c=0$ in \eqref{eq:thetastar0}. If non constant predictors have a zero mean, then $\tilde{\theta}_*^{(0)}$, like the other $\tilde{\theta}_*^{(j)}$, is given by the second part of \eqref{eq:thetastar}. This configuration is interesting because in practice, $\theta_j$ values that are derived from RL algorithms are not subject to the budget constraint (see, e.g., step 5 in Table \ref{tab:REINFORCE}). 

In addition, the fact that $\tilde{\theta}_*^{(j)}$ decreases to zero when $\sigma_{\beta,j}^2$ increases to infinity is consistent with the literature that finds that the EW portfolio is optimal under high model ambiguity (see, e.g., \cite{pflug2012} and \cite{maillet2015global}). Indeed if the non-constant loadings of $\bm{\beta}$ are subject to a very high level of uncertainty, then the investor will naturally and comparatively trust the constant factor much more. This results in an optimal allocation that is uniform across assets.

\subsection{Cross-sectional betas}

We illustrate the implications of Lemma \ref{lem:ass} by running monthly regressions to estimate the loadings in Equation \eqref{eq:model}. To ease interpretability, we restrict the analysis to the three most common factors in the literature: size (market capitalization), value (price-to-book) and momentum (12 month to 1 month return). The largest correlation between them is 0.18 on the whole sample, thus the hypothesis of diagonal covariance matrix is not too far-fetched. Each month, we report the OLS coefficients for Equation \eqref{eq:model}, where $\bm{r}$ is the vector of \textit{future} one month returns.

In the upper panel of Figure \ref{fig:loads}, we depict the estimated $\hat{\beta}_j$ for the three factors plus the constant. In addition, in the lower panel, we plot the scaled unconstrained theta values $\tilde{\theta}^{(j)}= \frac{\hat{\beta}_j/(2N)}{\hat{\sigma}_{X,j}^2\hat{\sigma}_{\beta,j}^2+\hat{\sigma}_\epsilon^2/N}$ (i.e., with $\gamma=2$, which is without loss of generality, as $\gamma$ is only a normalizing constant). One additional reason we resort to unconstrained $\theta^{(j)}$ values is that they do not depend on the risk aversion parameter, which only plays the role of a scaling factor.
 
\begin{figure}[!h]
\begin{center}
\includegraphics[width=13cm]{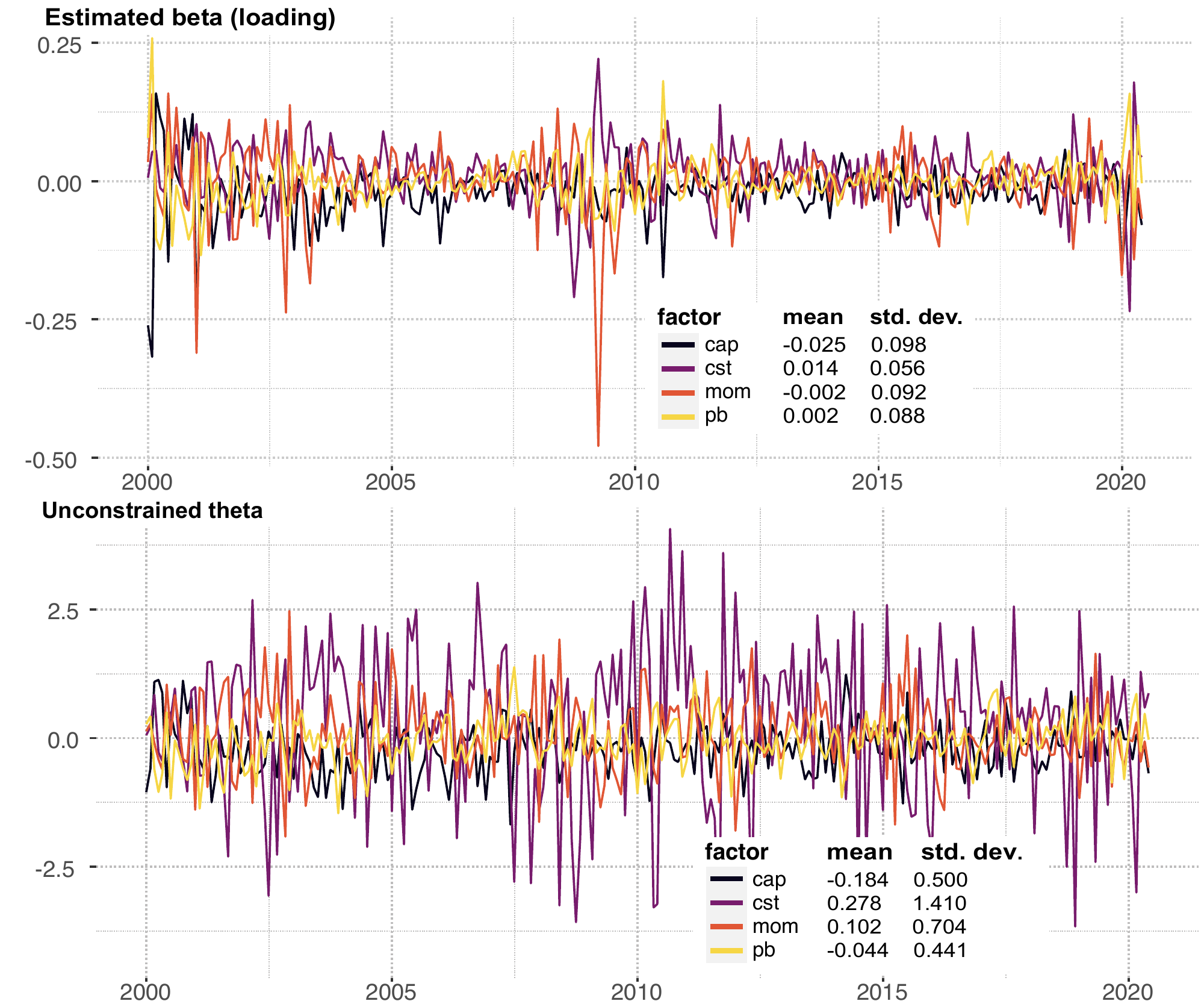}\vspace{-3mm}
\caption{\textbf{Panel betas and scaled unconstrained thetas}. \small We plot the estimated panel betas $\hat{\beta}_j$ for each month in the dataset (upper panel). For each of the four features (3 factors + constant), we show the scaled values $\tilde{\theta}^{(j)}= \frac{\hat{\beta}_j/(2N)}{\hat{\sigma}_{X,j}^2\hat{\sigma}_{\beta,j}^2+\hat{\sigma}_\epsilon^2/N}$ in the lower panel. Note: the risk aversion parameter (scaling constant) is $\gamma=2$. } 
\label{fig:loads}
\end{center}
\end{figure}

All betas and unconstrained thetas oscillate strongly, but their means and deviations are telling. The $\theta^{(0)}$ associated to the constant is volatile, but solidly positive on average, and by far dominating in magnitude, while the values for market capitalization are negative (which tends to be consistent with the so-called size anomaly).

\subsection{Comparing learning schemes}

Let us briefly summarize how agents allocate in the two frameworks (RL versus \textit{loadings}-based):
\begin{itemize}
\item When resorting to RL, the agent learns (via the policy gradient) by computing the \textbf{sensitivity} of the performance metric (average return or Sharpe ratio) with respect to variations in the parameters of the policy ($\bm{\theta}$).
\item In a more conventional characteristics-based asset pricing approach, the econometrician will evaluate the \textbf{exposure} of returns to firm specific attributes ($\bm{\beta}$). These exposures can then be translated into portfolio weights, when optimizing the average of a given utility function (see Lemma \ref{lem:ass} when the utility is quadratic). 
\end{itemize}


Theoretically, there are no reasons why these two approaches should be linked (they are hard to reconcile analytically, even though both seek to give more weight to assets that are expected to outperform - see Section \ref{sec:pgm}). However, empirically, we find some consistency between the two methods. One common feature is the dominance of the constant in the upper panel of Figure \ref{fig:theta} and in the lower panel of Figure \ref{fig:loads}. This indicates that both approaches find a strong common factor in the cross-section of returns which cannot be explained by firm level idiosyncasies. Heuristically, when the estimated $\hat{\beta}_0$ (which drives $\tilde{\theta}^{(0)}$) is high, returns are on average high in the cross-section, thus, the sensitivity of performance to variations in $\theta^{(0)}$ should also be positive. This is when the orange line in Figure \ref{fig:theta} is either at its ceiling, or increasing. When the policy learns over longer longer samples (bottom panel of Figure \ref{fig:theta}), then the long-term positivity of $\hat{\beta}_0$ (which is linked to that of the equity premium) pushes $\theta^{(0)}$ upwards.


We thus wish to further investigate the link between the two learning processes. On the one hand, the driving element in the RL allocation is the policy update $\Delta \bm{\theta}_t=\bm{\theta}_t-\bm{\theta}_{t-1}=\alpha \widehat{\nabla J(\mtrx{\theta}_t)} $ (see Equation \eqref{eq:pgrad}). On the other hand, we pick the optimal (unconstrained) $\tilde{\bm{\theta}}_t$ in Equation \eqref{eq:thetastar} to proxy for the information that is processed by the asset pricing model during the period between $t-1$ and $t$. In Figure \ref{fig:correl}, we plot the first values ($y$-axis) against the second ones ($x$-axis).

\begin{figure}[!h]
\begin{center}
\includegraphics[width=13cm]{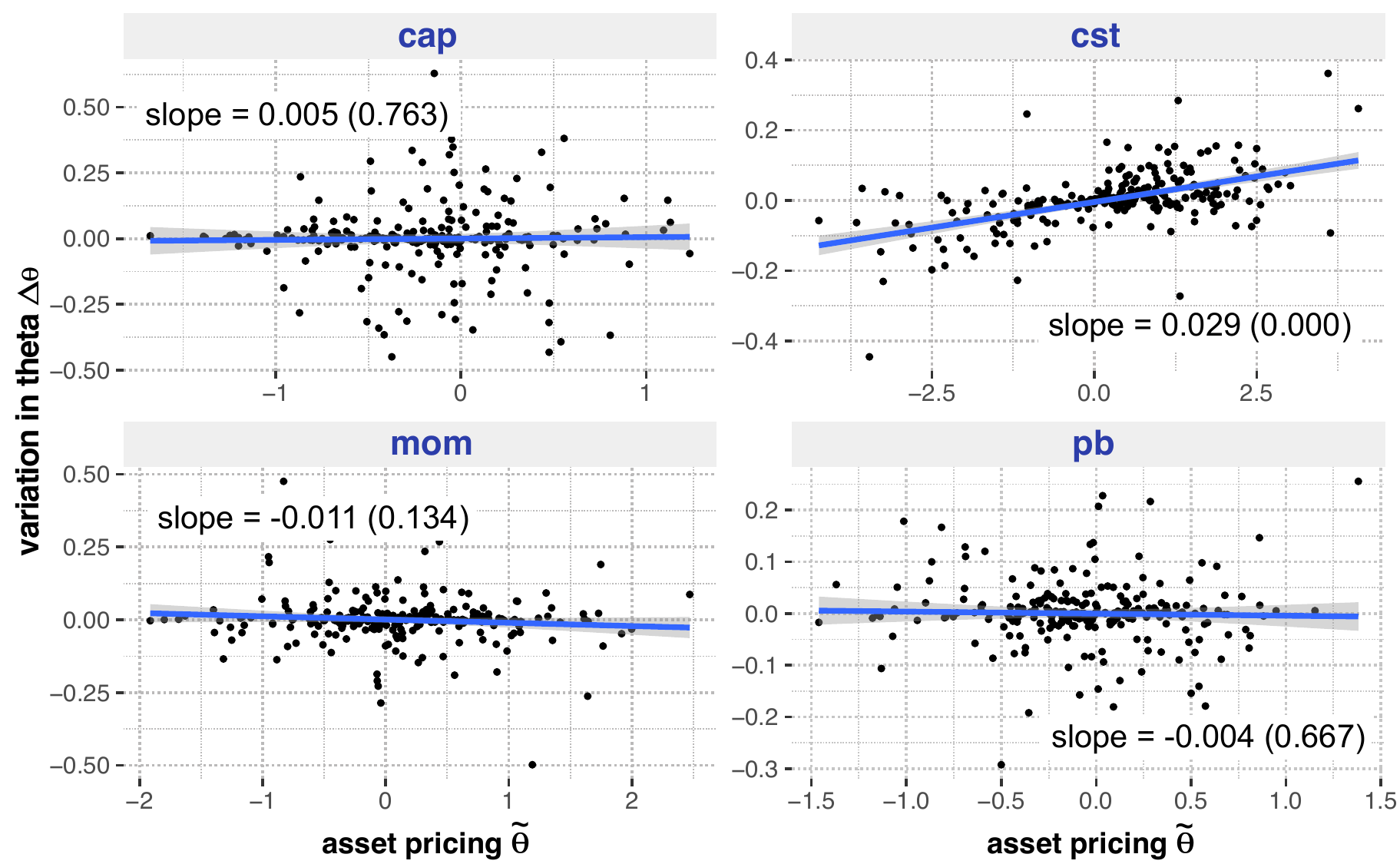}\vspace{-3mm}
\caption{\textbf{Policy gradients versus optimal asset pricing parameters}. \small We plot the variations in policy parameters $\Delta \bm{\theta}_t$ versus the unconstrained asset pricing based $\tilde{\bm{\theta}}_{*t}$ defined in Equation \eqref{eq:thetastar} and plotted in Figure \ref{fig:loads}. The former correspond to the first policy form (\textbf{F1}) updated from bootstrapped returns. The latter are built with returns from month $t-1$ to $t$ for consistency reasons (i.e., to match the informational set with which $\Delta \bm{\theta}_t$ is built). The slope of the fitted linear relationship is reported, along with the corresponding $p$-value (in brackets). Note: the risk aversion parameter (scaling constant) is $\gamma=2$. } 
\label{fig:correl}
\end{center}
\end{figure}

There is only one dimension for which the two schemes yield consistent values: the constant (upper right quadrant). On all purely firm-specific characteristics, the approaches seldom agree and the correlation between the two approaches is not significantly different from zero. Thus, apart from the relative agreement on the dominance of the common factor in the cross-section of stocks, it is hard to fully reconcile the two methods.


\section{Conclusion}
\label{sec:5}

In this article, we combine reinforcement learning with factor investing. The investor learns from the impact of firm-specific characteristics on a chosen performance measure. The technical machinery relies on tractable expressions derived from analytical properties of the Dirichlet distribution. This allows to keep allocation decision inside a simplex, which is the unique requirement in long-only portfolios.

Empirically, the approach yields weights that are very diversified, akin to the 1/N portfolio. One interpretation is that the learning process captures the importance of a common factor that drives the cross-section of stocks beyond their factorial idiosyncrasies. We compare the RL decisions to those of a simple asset pricing model and find significant differences for the non-constant characteristics. This shows the peculiarities stemming from the RL-based methodology. 

Interestingly, and in spite of a wide range of implementation choices, the fact that the RL portfolios remain in the vicinity of uniform allocations underlines the efficiency of the latter. To conclude, our results are less telling about the methodology than about the properties of the market data we use. They highlight the weak pricing ability of the traditional characteristics used in empirical asset pricing.

\begin{appendices}


\section{Some properties of the Dirichlet distribution}
\label{sec:PropDirichlet}

Let $\mtrx{W}$ be a vector with Dirichlet distribution, which we denote by $\mtrx{W} \sim \mathrm{Dir} \left(\mathbf{a}\right)$. The marginal distributions are Beta distributions: for $n=1,\dotsc,N$,
\[ W_{n} \sim \mathrm{Beta} \left(a_{n},\sigma-a_{n}\right) \]
and the two-dimensional marginal distributions are Dirichlet: for $1\le n<m\le N$,
\[ \begin{bmatrix}W_{n} & W_{m}\end{bmatrix}^\intercal \sim \mathrm{Dir}\left(a_{n},a_{m},\sigma-a_{n}-a_{m}\right)\]

Let $\digamma$ denotes the Digamma function, the derivative of the natural logarithm of the Gamma function. We have the following

\begin{proposition}\label{prop:wlnw}
Let $\mathbf{W} \sim \mathrm{Dir} \left(\mathbf{a}\right)$. Then for $1\le n\le N$,
\begin{align*}
\mathbb{E}\left[\ln W_{n}\right] &= \digamma\left(a_{n}\right)-\digamma\left(\sigma\right)\\
\mathbb{E}\left[W_{n}\ln W_{n}\right] &= \frac{a_{n}}{\sigma}\left(\digamma\left(a_{n}\right)+\frac{1}{a_{n}}-\digamma\left(\sigma\right)-\frac{1}{\sigma}\right)\\
\mathbb{E}\left[W_{n}\ln W_{m}\right] &= \frac{a_{n}}{\sigma}\left(\digamma\left(a_{m}\right)-\digamma\left(\sigma\right)-\frac{1}{\sigma}\right) \quad n\neq m
\end{align*}
\end{proposition}

Proofs are given below.

\subsection{Expectation of \texorpdfstring{$\ln W_{n}$}{ln Wn}}
We have $W_{n} \sim \mathrm{Beta}\left(a_{n},\sigma-a_{n}\right)$. Consequently,
\begin{align*}
\mathbb{E}\left[\ln W_{n}\right] &= \frac{1}{B\left(a_{n},\sigma-a_{n}\right)}\int_{0}^{1}\ln w_{n}w_{n}^{a_{n}-1}\left(1-w_{n}\right)^{\sigma-a_{n}-1}dw_{n}
&= \frac{1}{B\left(a_{n},\sigma-a_{n}\right)}\int_{0}^{1}\ln w_{n}\varphi(w_{n},a_{n})dw_{n},
\end{align*}
where
\[ \varphi(w_{n},a_{n})=w_{n}^{a_{n}-1}\left(1-w_{n}\right)^{\sigma-a_{n}-1}=e^{\left(a_{n}-1\right)\ln w_{n}+(\sigma-a_{n}-1)\ln\left(1-w_{n}\right)} . \]
As $a_{n}$ cancels out in $\sigma-a_{n}$,
\[ \frac{\partial}{\partial a_{n}}\varphi\left(w_{n},a_{n}\right)=\ln w_{n}\varphi\left(w_{n},a_{n}\right), \]
and
\begin{align*}
\mathbb{E}\left[\ln W_{n}\right]	&=\frac{1}{B\left(a_{n},\sigma-a_{n}\right)}\int_{0}^{1}\frac{\partial}{\partial a_{n}}\varphi\left(w_{n},a_{n}\right)dw_{n}\\
	&=\frac{1}{B\left(a_{n},\sigma-a_{n}\right)}\frac{\partial}{\partial a_{n}}\int_{0}^{1}\varphi\left(w_{n},a_{n}\right)dw_{n}\\
	&=\frac{1}{B\left(a_{n},\sigma-a_{n}\right)}\frac{\partial}{\partial a_{n}}B\left(a_{n},\sigma-a_{n}\right)\\
	&=\frac{\mathrm{d}}{\mathrm{d}a_{n}}\ln B\left(a_{n},\sigma-a_{n}\right)\\
	&=\frac{\mathrm{d}}{\mathrm{d}a_{n}}\left(\ln\Gamma(a_{n})+\ln\Gamma(\sigma-a_{n})-\ln\Gamma(\sigma)\right)\\
	&=\digamma(a_{n})-\digamma(\sigma).
\end{align*}

\subsection{Expectation of \texorpdfstring{$W_{n}\ln W_{n}$}{Wn ln Wn}}
We have $W_{n} \sim \mathrm{Beta}\left(a_{n},\sigma-a_{n}\right)$. Then,
\begin{align*}
\mathbb{E}\left[W_{n}\ln W_{n}\right]	&=\frac{1}{B\left(a_{n},\sigma-a_{n}\right)}\int_{0}^{1}w_{n}\ln w_{n}w_{n}^{a_{n}-1}\left(1-w_{n}\right)^{\sigma-a_{n}-1}dw_{n}\\
&=\frac{1}{B\left(a_{n},\sigma-a_{n}\right)}\int_{0}^{1}\ln w_{n}\varphi(w_{n},a_{n})dw_{n},
\end{align*}
where
\[ \varphi(w_{n},a_{n})=w_{n}^{a_{n}}\left(1-w_{n}\right)^{\sigma-a_{n}-1} .\]
In addition,
\begin{align*}
\mathbb{E}\left[W_{n}\ln W_{n}\right]	&=\frac{1}{B\left(a_{n},\sigma-a_{n}\right)}\int_{0}^{1}\frac{\partial}{\partial a_{n}}\varphi\left(w_{n},a_{n}\right)dw_{n}\\
&=\frac{1}{B\left(a_{n},\sigma-a_{n}\right)}\frac{\partial}{\partial a_{n}}\int_{0}^{1}\varphi\left(w_{n},a_{n}\right)dw_{n}\\
&=\frac{1}{B\left(a_{n},\sigma-a_{n}\right)}\frac{\partial}{\partial a_{n}}B\left(a_{n}+1,\sigma-a_{n}\right)\\
&=\frac{B\left(a_{n}+1,\sigma-a_{n}\right)}{B\left(a_{n},\sigma-a_{n}\right)}\frac{\mathrm{d}}{\mathrm{d}a_{n}}\ln B\left(a_{n}+1,\sigma-a_{n}\right).
\end{align*}

The ratio of Betas simplifies to
\[ \frac{B\left(a_{n}+1,\sigma-a_{n}\right)}{B\left(a_{n},\sigma-a_{n}\right)} =\frac{\Gamma\left(a_{n}+1\right)\Gamma\left(\sigma-a_{n}\right)\Gamma\left(\sigma\right)}{\Gamma\left(\sigma+1\right)\Gamma\left(a_{n}\right)\Gamma\left(\sigma-a_{n}\right)}
=\frac{a_{n}\Gamma\left(a_{n}\right)\Gamma\left(\sigma\right)}{\sigma\Gamma\left(\sigma\right)\Gamma\left(a_{n}\right)}
=\frac{a_{n}}{\sigma} .\]

Therefore we can conclude
\begin{align*}
\mathbb{E}\left[W_{n}\ln W_{n}\right]	&=\frac{a_{n}}{\sigma}\frac{\mathrm{d}}{\mathrm{d}a_{n}}\ln B\left(a_{n}+1,\sigma-a_{n}\right)\\
&=\frac{a_{n}}{\sigma}\frac{\mathrm{d}}{\mathrm{d}a_{n}}\left(\ln\Gamma(a_{n}+1)+\ln\Gamma(\sigma-a_{n})-\ln\Gamma(\sigma+1)\right)\\
&=\frac{a_{n}}{\sigma}\left(\digamma(a_{n}+1)-\digamma(\sigma+1)\right)\\
&=\frac{a_{n}}{\sigma}\left(\digamma(a_{n})+\frac{1}{a_{n}}-\digamma(\sigma)-\frac{1}{\sigma}\right).
\end{align*}

\subsection{Expectation of \texorpdfstring{$W_{n}\ln W_{m}$}{Wn ln Wm}}
We have $\begin{bmatrix}W_{n} & W_{m}\end{bmatrix}\sim\mathrm{Dir}\left(\mtrx{a}_{n,m}\right)$, where $\mtrx{a}_{n,m}=\begin{bmatrix}a_{n} & a_{m} & \sigma-a_{n}-a_{m}\end{bmatrix}$. Also,
\begin{align*}
\mathbb{E}\left[W_{n}\ln W_{m}\right] &=\frac{1}{B\left(\mtrx{a}_{n,m}\right)}\int_{0}^{1}\int_{0}^{1-w_{n}}w_{n}\ln w_{m}w_{n}^{a_{n}-1}w_{m}^{a_{m}-1}\left(1-w_{n}-w_{m}\right)^{\sigma-a_{n}-a_{m}-1}dw_{n}dw_{m}\\
&=\frac{1}{B\left(\mtrx{a}_{n,m}\right)}\int_{0}^{1}w_{n}^{a_{n}}\left(\int_{0}^{1-w_{n}}\ln w_{m}w_{m}^{a_{m}-1}\left(1-w_{n}-w_{m}\right)^{\sigma-a_{n}-a_{m}-1}dw_{m}\right)dw_{n}.
\end{align*}
The inner integral is
\[ I=\int_{0}^{\lambda}\ln ww^{a_{m}-1}\left(\lambda-w\right)^{\sigma-a_{n}-a_{m}-1}dw .\]
With the change of variable $\lambda t=w$
\begin{align*}
I &= \int_{0}^{1}\left(\ln\lambda+\ln t\right)\lambda^{a_{m}-1}t^{a_{m}-1}\lambda^{\sigma-a_{n}-a_{m}-1}\left(1-t\right)^{\sigma-a_{n}-a_{m}-1}\lambda dt\\
&= \lambda^{\sigma-a_{n}-1}\left(\ln\lambda\int_{0}^{1}t^{a_{m}-1}\left(1-t\right)^{\sigma-a_{n}-a_{m}-1}dt+\int_{0}^{1}\ln t \ t^{a_{m}-1}\left(1-t\right)^{\sigma-a_{n}-a_{m}-1}dt\right)\\
&= \lambda^{\sigma-a_{n}-1}\left(\ln\lambda B\left(a_{m},\sigma-a_{n}-a_{m}\right)+\int_{0}^{1}\ln t\varphi\left(t,a_{m}\right)dt\right),
\end{align*}
where
\[ \varphi\left(t,a_{m}\right)	=t^{a_{m}-1}\left(1-t\right)^{\sigma-a_{n}-a_{m}-1} .\]
As $a_{m}$ cancels out in $\sigma-a_{n}-a_{m}$,
\[ \frac{\partial}{\partial a_{m}}\varphi\left(t,a_{m}\right)=\ln t\varphi\left(t,a_{m}\right) .\]
Therefore,
\begin{align*}
\int_{0}^{1}\ln t\varphi\left(t,a_{m}\right)dt &=\int_{0}^{1}\frac{\partial}{\partial a_{m}}\varphi\left(t,a_{m}\right)dt=\frac{\partial}{\partial a_{m}}\int_{0}^{1}\varphi\left(t,a_{m}\right)dt \\
&= \frac{\mathrm{d}}{\mathrm{d}a_{m}}B\left(a_{m},\sigma-a_{n}-a_{m}\right) \\
&= B\left(a_{m},\sigma-a_{n}-a_{m}\right)\frac{\mathrm{d}}{\mathrm{d}a_{m}}\ln B\left(a_{m},\sigma-a_{n}-a_{m}\right),
\end{align*}
and
\begin{align*}
\frac{\mathrm{d}}{\mathrm{d}a_{m}}\ln B\left(a_{m},\sigma-a_{n}-a_{m}\right) &= \frac{\mathrm{d}}{\mathrm{d}a_{m}}\left(\ln\Gamma\left(a_{m}\right)+\ln\Gamma\left(\sigma-a_{n}-a_{m}\right)-\ln\Gamma\left(\sigma-a_{n}\right)\right)\\
&= \digamma\left(a_{m}\right)-\digamma\left(\sigma-a_{n}\right).
\end{align*}
This gives
\[ I=\left(1-w_{n}\right)^{\sigma-a_{n}-1}B\left(a_{m},\sigma-a_{n}-a_{m}\right)\left(\ln\left(1-w_{n}\right)+\digamma\left(a_{m}\right)-\digamma\left(\sigma-a_{n}\right)\right) .\]
Back to the expectation,
\[ \mathbb{E}\left[W_{n}\ln W_{m}\right]=\frac{B\left(a_{m},\sigma-a_{n}-a_{m}\right)}{B\left(\mathbf{a}_{n,m}\right)}\int_{0}^{1}w_{n}^{a_{n}}\left(1-w_{n}\right)^{\sigma-a_{n}-1}\left(\ln\left(1-w_{n}\right)+\digamma\left(a_{m}\right)-\digamma\left(\sigma-a_{n}\right)\right)dw_{n} .\]
The ratio of Betas simplifies to
\begin{align*}
\frac{B\left(a_{m},\sigma-a_{n}-a_{m}\right)}{B\left(\mathbf{a}_{n,m}\right)}	&=\frac{\Gamma\left(a_{m}\right)\Gamma\left(\sigma-a_{n}-a_{m}\right)\Gamma\left(\sigma\right)}{\Gamma\left(\sigma-a_{n}\right)\Gamma\left(a_{n}\right)\Gamma\left(a_{m}\right)\Gamma\left(\sigma-a_{n}-a_{m}\right)} \\
&= \frac{\Gamma\left(\sigma\right)}{\Gamma\left(\sigma-a_{n}\right)\Gamma\left(a_{n}\right)} =\frac{1}{B\left(a_{n},\sigma-a_{n}\right)},
\end{align*}
and the integral splits in two parts which are
\begin{align*}
I_{1}	&= \left(\digamma\left(a_{m}\right)-\digamma\left(\sigma-a_{n}\right)\right)\int_{0}^{1}w_{n}^{a_{n}}\left(1-w_{n}\right)^{\sigma-a_{n}-1}dw_{n}\\
&= \left(\digamma\left(a_{m}\right)-\digamma\left(\sigma-a_{n}\right)\right)B(a_{n}+1,\sigma-a_{n})
\end{align*}
and
\begin{align*}
I_{2}	&= \int_{0}^{1}\ln\left(1-w_{n}\right)w_{n}^{a_{n}}\left(1-w_{n}\right)^{\sigma-a_{n}-1}dw_{n}\\
&= \int_{0}^{1}\frac{\partial}{\partial a_{k}}\left(w_{n}^{a_{n}}\left(1-w_{n}\right)^{\sigma-a_{n}-1}\right)dw_{n}
\end{align*}
for any $1\le k\le N$, $k\neq n$. Thus,
\begin{align*}
I_{2}	&= \frac{\partial}{\partial a_{k}}\left(\int_{0}^{1}w_{n}^{a_{n}}\left(1-w_{n}\right)^{\sigma-a_{n}-1}dw_{n}\right)\\
&= B\left(a_{n}+1,\sigma-a_{n}\right)\frac{\mathrm{d}}{\mathrm{d}a_{k}}\ln B\left(a_{n}+1,\sigma-a_{n}\right)\\
&= B\left(a_{n}+1,\sigma-a_{n}\right)\left(\digamma\left(\sigma-a_{n}\right)-\digamma\left(\sigma+1\right)\right).
\end{align*}

Putting everything together
\[ \mathbb{E}\left[W_{n}\ln W_{m}\right]=\frac{B(a_{n}+1,\sigma-a_{n})}{B\left(a_{n},\sigma-a_{n}\right)}\left(\digamma\left(a_{m}\right)-\digamma\left(\sigma-a_{n}\right)+\digamma\left(\sigma-a_{n}\right)-\digamma\left(\sigma+1\right)\right) .\]
Again, the ratio of Betas simplifies to
\begin{equation*}
\frac{B(a_{n}+1,\sigma-a_{n})}{B\left(a_{n},\sigma-a_{n}\right)}	= \frac{\Gamma\left(a_{n}+1\right)\Gamma\left(\sigma-a_{n}\right)\Gamma\left(\sigma\right)}{\Gamma\left(\sigma+1\right)\Gamma\left(a_{n}\right)\Gamma\left(\sigma-a_{n}\right)}
= \frac{a_{n}\Gamma\left(a_{n}\right)\Gamma\left(\sigma\right)}{\sigma\Gamma\left(\sigma\right)\Gamma\left(a_{n}\right)} = \frac{a_{n}}{\sigma}.
\end{equation*}
Finally
\[ \mathbb{E}\left[W_{n}\ln W_{m}\right]=\frac{a_{n}}{\sigma}\left(\digamma\left(a_{m}\right)-\digamma\left(\sigma\right)-\frac{1}{\sigma}\right) .\]


\section{Link between the asset characteristics and the portfolio composition}
\label{sec:link}

\subsection{Rewriting the concentration parameters for policy F1}
We introduce the notation
\[ \tilde{w}_{t,n} = \frac{1}{N\theta_{t}^{(0)}}\sum_{k=1}^K \theta_{t}^{(k)}x_{t,n}^{(k)} ,\]
so that
\[ a_{t,n} = (\mtrx{x}_{t,n})^{\intercal}\mtrx{\theta}_{t} = \theta_{t}^{(0)} + N\theta_{t}^{(0)} \tilde{w}_{t,n} .\]
We have
\[ \sum_{n=1}^{N} \tilde{w}_{t,n} = \frac{1}{N\theta_{t}^{(0)}} \sum_{k=1}^K \theta_{t}^{(k)} \sum_{n=1}^{N} x_{t,n}^{(k)}, \]
but, due to the processing of the features described in section~\ref{sec:data} (at each time $t$, each feature $k$ is uniformly distributed over the $[-0.5,0.5]$ interval) across firms), it holds that $\sum_{n=1}^{N} x_{t,n}^{(k)} = 0$ hence $\sum_{n=1}^{N} \tilde{w}_{t,n} = 0$.
From which we obtain
\begin{align}
\sigma_{t} &= N\theta_{t}^{(0)},\\
\frac{a_{t,n}}{\sigma_{t}} &= \frac{1}{N} + \tilde{w}_{t,n} .
\end{align}
Finally, we note for a later use that
\begin{equation}\label{eq:majorant}
\abs{\tilde{w}_{t,n}} \le \frac{1}{N \abs{\theta_{t}^{(0)}}} \sum_{k=1}^K \abs{\theta_{t}^{(k)}} \abs{x_{t,n}^{(k)}} \le \frac{1}{2N} \frac{\norm{\mtrx{\theta}^{(-0)}_t}_1}{\abs{\theta_{t}^{(0)}}},
\end{equation}
where $\mtrx{\theta}^{(-0)}_t = [\theta_{t}^{(1)}, \dotsc, \theta_{t}^{(K)}]$.

\subsection{Rewriting the concentration parameters for policy F2}
Using the same notations as above
\[ a_{t,n} = e^{\theta_{t}^{(0)}} e^{N\theta_{t}^{(0)}\tilde{w}_{t,n}} \quad \sigma_{t} = e^{\theta_{t}^{(0)}} \sum_{m=1}^{N} e^{N\theta_{t}^{(0)}\tilde{w}_{t,m}} ,\]
so that
\begin{equation*}
\frac{a_{t,n}}{\sigma_{t}} = \left( \sum_{m=1}^{N} e^{N\theta_{t}^{(0)} (\tilde{w}_{t,m}-\tilde{w}_{t,n})} \right)^{-1}
\end{equation*}
Furthermore, from equation~\eqref{eq:majorant} we get
\begin{equation*}
\abs{N\theta_{t}^{(0)} (\tilde{w}_{t,m}-\tilde{w}_{t,n})} \le \norm{\mtrx{\theta}^{(-0)}_t}_1,
\end{equation*}
and hence, if $\norm{\mtrx{\theta}^{(-0)}_t}_1 \ll 1$, we obtain
\begin{equation*}
\frac{a_{t,n}}{\sigma_{t}} \approx \left( \sum_{m=1}^{N} (1 + N\theta_{t}^{(0)} (\tilde{w}_{t,m}-\tilde{w}_{t,n})) \right)^{-1} = \left( N (1 - N\theta_{t}^{(0)} \tilde{w}_{t,n}) \right)^{-1},
\end{equation*}
that is
\begin{equation}
\frac{a_{t,n}}{\sigma_{t}} \approx \frac{1}{N} + \theta_{t}^{(0)} \tilde{w}_{t,n},
\end{equation}
and also
\begin{equation}
\sigma_{t} \approx N e^{\theta_{t}^{(0)}}
\end{equation}

\subsection{Approximations when the bias is large compared to the other weights}
We see that, for policy $\mathbf{F1}$ when $\abs{\theta_{t}^{(0)}} \gg \norm{\mtrx{\theta}^{(-0)}_t}_1$, and for policy $\mathbf{F2}$ when $1 \gg \norm{\mtrx{\theta}^{(-0)}_t}_1$, then
\begin{equation}
\frac{a_{t,n}}{\sigma_{t}} \approx \frac{1}{N}.
\end{equation}


\section{Proofs}

\subsection{Proof of Proposition \ref{prop:PolicyValue}} \label{sec:Proof-PolicyValue}

We first compute the future periods expected returns given some state $S_{t} = (\rho_t, \mtrx{X}_{t})$ and action $A_{t}= \mtrx{w}_{t}$.

\begin{lemma}\label{lem:ExpRet}
\begin{align*}
\mathbb{E}_{\mtrx{\theta}}[ \rho_{t+l} \mid (\rho_t, \mtrx{X}_{t}), \mtrx{w}_{t}] = 
\begin{cases}
\mtrx{w}_{t}^\intercal f \left(\mtrx{X}_{t}\right) &l=1\\
\int_{\mathcal{M}} \mathbb{E}_{\mtrx{\theta}} \left[ \mtrx{w} \mid \xi \right]^\intercal f\left(\xi\right) \mathbb{P}_{t}^{t+l-1}(d\xi \mid \mtrx{X}_{t}) &l \ge 2
\end{cases},
\end{align*}
where $\mathbb{E}_{\mtrx{\theta}} \left[ \mtrx{w} \mid \xi \right]$ is the vector expectation of the portfolio composition under the stochastic policy $\pi_{\mtrx{\theta}}$.
\end{lemma}

\begin{proof}
Case $l=1$. We have
\begin{align*}
\mathbb{E}_{\mtrx{\theta}}[ \rho_{t+1} \mid (\rho_t, \mtrx{X}_{t}), \mtrx{w}_{t}] &= \mathbb{E}_{\mtrx{\theta}}[ \mtrx{w}_{t}^\intercal (f(\mtrx{X}_{t}) + \mtrx{\epsilon}_{t+1}) \mid (\rho_t, \mtrx{X}_{t}), \mtrx{w}_{t}]\\
&= \mtrx{w}_{t}^\intercal f(\mtrx{X}_{t}) + \mtrx{w}_{t}^\intercal \mathbb{E}_{\mtrx{\theta}}[\mtrx{\epsilon}_{t+1}],
\end{align*}
which gives the result because $\mtrx{\epsilon}_{t+1}$ is a zero mean White Noise.

For the case $l \ge 2$, we use the tower property of the conditional expectation. Let $\mathscr{F}_t$ be the sigma-algebra generated by $\{S_{t}, A_t, S_{t-1}, A_{t-1}, \dotsc \}$. By the Markov property of our setting, we know that $\mathbb{E}_{\mtrx{\theta}}[ \rho_{t+l} \mid (\rho_t, \mtrx{X}_{t}), \mtrx{w}_{t}] = \mathbb{E}_{\mtrx{\theta}}[ \rho_{t+l} \mid \mathscr{F}_t]$. Then
\begin{equation*}
\mathbb{E}_{\mtrx{\theta}}[ \rho_{t+l} \mid \mathscr{F}_t] = \mathbb{E}_{\mtrx{\theta}}[ \mathbb{E}_{\mtrx{\theta}}[ \rho_{t+l} \mid \mathscr{F}_{t+l-1}] \mid \mathscr{F}_t],
\end{equation*}
where the inner conditional expectation is given by the case $l=1$. Then
\begin{equation*}
\mathbb{E}_{\mtrx{\theta}}[ \rho_{t+l} \mid \mathscr{F}_t] = \mathbb{E}_{\mtrx{\theta}}[ \mtrx{w}_{t+l-1}^\intercal f \left(\mtrx{X}_{t+l-1}\right) \mid \mathscr{F}_t] = \mathbb{E}_{\mtrx{\theta}}[ \mathbb{E}_{\mtrx{\theta}}[ \mtrx{w}_{t+l-1}^\intercal f \left(\mtrx{X}_{t+l-1}\right) \mid \mathscr{F}_{t+l-2}] \mid \mathscr{F}_t].
\end{equation*}
The inner conditional expectation is, using equation~\eqref{eq:TransitionProba},
\begin{multline*}
\mathbb{E}_{\mtrx{\theta}}[ \mtrx{w}_{t+l-1}^\intercal f \left(\mtrx{X}_{t+l-1}\right) \mid \mathscr{F}_{t+l-2}] =\\
 \int_{\mathcal{M}} \int_{\mathbb{R}} \int_{\Delta} \mtrx{w}^\intercal f (\xi) \pi_{\mtrx{\theta}} \left( d \mtrx{w} \mid \xi \right) 
\bm{P}_\epsilon (\mathbb{T}^{-1}(dr \mid \mtrx{X}_{t+l-2}, \mtrx{w}_{t+l-2})) \mathbb{P}_{t+l-2}^{t+l-1} \left( d\xi \mid \mtrx{X}_{t+l-2} \right),
\end{multline*}
 which simplifies to
\begin{equation*}
\mathbb{E}_{\mtrx{\theta}}[ \mtrx{w}_{t+l-1}^\intercal f \left(\mtrx{X}_{t+l-1}\right) \mid \mathscr{F}_{t+l-2}] =
 \int_{\mathcal{M}} \mathbb{E}_{\mtrx{\theta}} \left[ \mtrx{w} \mid \xi \right]^\intercal f\left(\xi\right) \mathbb{P}_{t+l-2}^{t+l-1} \left( d\xi \mid \mtrx{X}_{t+l-2} \right).
\end{equation*}
If necessary, another application of the tower property is seen to lead to integrating over possible values for the state $\mtrx{X}_{t+l-2}$ given the state $\mtrx{X}_{t+l-3}$. Using the composition of the Markov transition probabilities as many times as required to reach the known state $\mtrx{X}_{t}$ leads to the result
\begin{equation*}
\mathbb{E}_{\mtrx{\theta}}[ \rho_{t+l} \mid \mathscr{F}_t] =  \int_{\mathcal{M}} \mathbb{E}_{\mtrx{\theta}} \left[ \mtrx{w} \mid \xi \right]^\intercal f\left(\xi\right) \mathbb{P}_{t}^{t+l-1} \left( d\xi \mid \mtrx{X}_{t} \right).
\end{equation*}
\end{proof}

We can now prove Proposition \ref{prop:PolicyValue}. For the risk-insensitive agent, $R_{t} = \rho_{t}$, so that 
\begin{equation*}
\mathbb{E}_{\mtrx{\theta}} \left[ R_{t+l} \mid S_{t}=(\rho_t, \mtrx{X}_{t}) \right] = \int_{\Delta} \left( \int_{\mathcal{S}} \rho_{t+l} \mathrm{Prob}\left( ds \mid S_{t}, \mtrx{w}_{t} \right) \right) \pi_{\mtrx{\theta}} \left( d\mtrx{w}_{t} \mid \mtrx{X}_{t}\right).
\end{equation*}
Lemma~\ref{lem:ExpRet} gives the value of the inner integral and we have two cases. If $l=1$, then
\begin{equation*}
\mathbb{E}_{\mtrx{\theta}} \left[ R_{t+1} \mid S_{t} \right]
= \int_{\Delta} \mtrx{w}_{t}^\intercal f \left(\mtrx{X}_{t}\right) \pi_{\mtrx{\theta}} \left( d\mtrx{w}_{t} \mid \mtrx{X}_{t}\right) 
= \left( \int_{\Delta} \mtrx{w}_{t} \pi_{\mtrx{\theta}} \left( d\mtrx{w}_{t} \mid \mtrx{X}_{t} \right) \right)^\intercal f \left(\mtrx{X}_{t}\right).
\end{equation*}
If $l \ge 2$, we have
\[ \mathbb{E}_{\mtrx{\theta}} \left[ R_{t+l}\mid S_{t} \right] = \int_{\Delta} \left( \int_{\mathcal{M}} \mathbb{E}_{\mtrx{\theta}} \left[ \mtrx{w} \mid \xi \right]^\intercal f\left(\xi\right)\mathbb{P}_{t}^{t+l-1}(d\xi\mid\mtrx{X}_{t}) \right) \pi_{\mtrx{\theta}} \left( d\mtrx{w}_{t} \mid \mtrx{X}_{t}\right). \]
The inner integral is not a function of the action $A_t = \mtrx{w}_{t}$ hence
\[ \mathbb{E}_{\mtrx{\theta}} \left[ R_{t+l}\mid S_{t} \right] = \int_{\mathcal{M}} \mathbb{E}_{\mtrx{\theta}} \left[ \mtrx{w} \mid \xi \right]^\intercal f\left(\xi\right)\mathbb{P}_{t}^{t+l-1}(d\xi\mid\mtrx{X}_{t}).\]

Now we show that the policy value takes a recursive form. First rewrite the policy value for the risk insensitive agent as
\begin{equation*}
V^{\mtrx{\theta}}(t,S_{t}) = \sum_{l=1}^{T-t} \mathbb{E}_{\mtrx{\theta}}\left[ R_{t+l} \mid S_{t} \right] = \mathbb{E}_{\mtrx{\theta}} \left[ R_{t+1} \mid S_{t} \right] + \sum_{l=1}^{T-(t+1)} \mathbb{E}_{\mtrx{\theta}}\left[ R_{t+1+l} \mid S_{t} \right] .
\end{equation*}
From the above result 
\begin{align*}
\mathbb{E}_{\mtrx{\theta}}\left[ R_{t+1+l} \mid S_{t} \right] &= \int_{\mathcal{M}} \mathbb{E}_{\mtrx{\theta}} \left[ \mtrx{w} \mid \xi \right]^\intercal f\left(\xi\right)\mathbb{P}_{t}^{t+l}(d\xi\mid\mtrx{X}_{t})\\ 
&= \int_{\mathcal{M}} \int_{\mathcal{M}} \mathbb{E}_{\mtrx{\theta}} \left[ \mtrx{w} \mid \xi \right]^\intercal f\left(\xi\right) \mathbb{P}_{t+1}^{t+l}(d\xi\mid\xi') \mathbb{P}_{t}^{t+1}(d\xi'\mid\mtrx{X}_{t})\\
&= \int_{\mathcal{M}} \mathbb{E}_{\mtrx{\theta}}\left[ R_{t+1+l} \mid S_{t+1} = \xi' \right] \mathbb{P}_{t}^{t+1}(d\xi'\mid\mtrx{X}_{t}) .
\end{align*}
Thus, the rightmost part of the policy value expression can be written as
\begin{equation*}
\sum_{l=1}^{T-(t+1)} \mathbb{E}_{\mtrx{\theta}}\left[ R_{t+1+l} \mid S_{t} \right] = \int_{\mathcal{M}} \sum_{l=1}^{T-(t+1)} \mathbb{E}_{\mtrx{\theta}}\left[ R_{t+1+l} \mid S_{t+1} = \xi' \right] \mathbb{P}_{t}^{t+1}(d\xi'\mid\mtrx{X}_{t}),
\end{equation*}
where the inner sum is $V^{\mtrx{\theta}}(t+1,\xi')$. Hence the result
\begin{equation*}
V^{\mtrx{\theta}}(t,S_{t}) = \mathbb{E}_{\mtrx{\theta}} \left[ R_{t+1} \mid S_{t} \right] + \int_{\mathcal{M}} V^{\mtrx{\theta}}(t+1,\xi) \mathbb{P}_{t}^{t+1}(d\xi\mid\mtrx{X}_{t}).
\end{equation*}

\subsection{Proof of Proposition \ref{prop:DirichletGrad}} \label{sec:Proof-DirichletGrad}
From equation \eqref{eq:dirichlet} we obtain
\begin{equation*}
\ln\pi\left(\mtrx{w}_{t}\mid \mtrx{X}_t, \mtrx{\theta}_{t}\right) = \ln \Gamma \left( \sigma_{t} \right) - \sum_{n=1}^N \ln \Gamma(a_{t,n}) + \sum_{n=1}^N (a_{t,n}-1) \ln w_{t,n}.
\end{equation*}
Then,
\begin{align*}
\nabla \ln \pi\left(\mtrx{w}_{t}\mid \mtrx{X}_t, \mtrx{\theta}_{t}\right) &= \digamma \left( \sigma_{t} \right) \sum_{n=1}^N \nabla a_{t,n} - \sum_{n=1}^N \digamma(a_{t,n}) \nabla a_{t,n} + \sum_{n=1}^N \ln w_n \nabla a_{t,n}\\
&= \sum_{n=1}^N \left( \digamma \left( \sigma_{t} \right) - \digamma(a_{t,n}) + \ln w_n \right) \nabla a_{t,n},
\end{align*}
where the gradients of the concentration parameters  $a_{t,n}$ are computed for both case given in equation~\eqref{eq:policies}.

\subsection{Proof of Proposition \ref{prop:PolicyGrad}} \label{sec:Proof-PolicyGrad}
We have
\begin{align*}
\nabla J\left(\mtrx{\theta}_{t}\right) &= \mathbb{E}_{\pi} \left[ G_t \nabla\ln\pi\left(\mtrx{w}_{t}\mid \mtrx{X}_t,\mtrx{\theta}_{t}\right)\mid \mtrx{X}_t,\mtrx{\theta}_{t} \right]\\
&= \sum_{l=1}^{T-t} \mathbb{E}_{\pi} \left[ \rho_{t+l} \nabla\ln\pi\left(\mtrx{w}_{t}\mid \mtrx{X}_t,\mtrx{\theta}_{t}\right)\mid \mtrx{X}_t,\mtrx{\theta}_{t} \right]\\
&= \sum_{l=1}^{T-t} \int_{\Delta} \left( \int_{\mathcal{S}} \rho_{t+l} \nabla\ln\pi\left(\mtrx{w}_{t}\mid \mtrx{X}_t,\mtrx{\theta}_{t}\right) \mathrm{Prob}\left( ds \mid S_{t}, \mtrx{w}_{t} \right) \right) \pi_{\mtrx{\theta}} \left( d\mtrx{w}_{t} \mid \mtrx{X}_{t}\right)\\
&= \sum_{l=1}^{T-t} \int_{\Delta} \left( \int_{\mathcal{S}} \rho_{t+l} \mathrm{Prob}\left( ds \mid S_{t}, \mtrx{w}_{t} \right) \right) \nabla\ln\pi\left(\mtrx{w}_{t}\mid \mtrx{X}_t,\mtrx{\theta}_{t}\right) \pi_{\mtrx{\theta}} \left( d\mtrx{w}_{t} \mid \mtrx{X}_{t}\right).
\end{align*}
The inner integrals are given by Lemma~\ref{lem:ExpRet} and
\begin{align*}
\nabla J\left(\mtrx{\theta}_{t}\right) &= \int_{\Delta} \mtrx{w}_{t}^\intercal f \left(\mtrx{X}_{t}\right) \nabla\ln\pi\left(\mtrx{w}_{t}\mid \mtrx{X}_t,\mtrx{\theta}_{t}\right) \pi_{\mtrx{\theta}} \left( d\mtrx{w}_{t} \mid \mtrx{X}_{t}\right)\\
&\quad + \sum_{l=2}^{T-t} \int_{\Delta} \underbrace{\left( \int_{\mathcal{M}} \mathbb{E}_{\mtrx{\theta}} \left[ \mtrx{w} \mid \xi \right]^\intercal f\left(\xi\right) \mathbb{P}_{t}^{t+l-1}(d\xi\mid\mtrx{X}_{t}) \right)}_{(\ast)} \nabla\ln\pi\left(\mtrx{w}_{t}\mid \mtrx{X}_t,\mtrx{\theta}_{t}\right) \pi_{\mtrx{\theta}} \left( d\mtrx{w}_{t} \mid \mtrx{X}_{t}\right)
\end{align*}
as $(\ast)$ is independent of the choice of $\mtrx{w}_{t}$, we can write the second part of the expression as
\[ \left( \int_{\Delta} \nabla\ln\pi\left(\mtrx{w}_{t}\mid \mtrx{X}_t,\mtrx{\theta}_{t}\right) \pi_{\mtrx{\theta}} \left( d\mtrx{w}_{t} \mid \mtrx{X}_{t}\right) \right) \sum_{l=2}^{T-t} \int_{\mathcal{M}} \mathbb{E}_{\mtrx{\theta}} \left[ \mtrx{w} \mid \xi \right]^\intercal f\left(\xi\right)\mathbb{P}_{t}^{t+l-1}(d\xi\mid\mtrx{X}_{t}).\]

From equations~\eqref{eq:DirichletGrad} and~\eqref{eq:logMean} we obtain that
\[ \int_{\Delta} \nabla\ln\pi\left(\mtrx{w}_{t}\mid \mtrx{X}_t,\mtrx{\theta}_{t}\right) \pi_{\mtrx{\theta}} \left( d\mtrx{w}_{t} \mid \mtrx{X}_{t}\right) = 0.\]
Therefore,
\[ \nabla J\left(\mtrx{\theta}_{t}\right) = \mathbb{E}_{\pi} \left[ \mtrx{w}_{t}^\intercal f \left(\mtrx{X}_{t}\right) \nabla\ln\pi\left(\mtrx{w}_{t}\mid \mtrx{X}_t,\mtrx{\theta}_{t}\right) \mid \mtrx{X}_t,\mtrx{\theta}_{t} \right] .\]

Let $f_{t,n}$ denotes the $n$-th element of the vector $f\left(\boldsymbol{X}_{t}\right)$. Then,
\begin{align*}
\nabla J\left(\mtrx{\theta}_{t}\right)
&=\mathbb{E}_{\pi}\left[\left(\sum_{m=1}^{N}w_{t,m}f_{t,m}\right)\left(\sum_{n=1}^{N}\left(\digamma\left(\sigma_{t}\right)-\digamma\left(a_{t,n}\right)+\ln w_{t,n}\right)\nabla a_{t,n}\right)\mid \mtrx{X}_{t},\mtrx{\theta}_{t}\right]\\
&=\sum_{n=1}^{N} \mathbb{E}_{\pi}\left[\left(\sum_{m=1}^{N}w_{t,m}f_{t,m}\right)\left(\digamma\left(\sigma_{t}\right)-\digamma\left(a_{t,n}\right)+\ln w_{t,n}\right)\mid \mtrx{X}_{t},\mtrx{\theta}_{t}\right] \nabla a_{t,n}\\
&=\sum_{n=1}^{N}\left(\digamma\left(\sigma_{t}\right)-\digamma\left(a_{t,n}\right)\right) \left(\sum_{m=1}^{N} \mathbb{E}_{\pi}\left[w_{t,m} \mid \mtrx{X}_{t},\mtrx{\theta}_{t}\right] f_{t,m} \right) \nabla a_{t,n}\\
&\quad + \sum_{n=1}^{N} \sum_{m=1}^{N} \mathbb{E}_{\pi}\left[w_{t,m} \ln w_{t,n}\mid \mtrx{X}_{t},\mtrx{\theta}_{t}\right] f_{t,m} \nabla a_{t,n}.
\end{align*}
On the one hand we have, by the expectation of a Dirichlet distributed random vector $\mathbb{E}_{\pi}[w_{t,m} \mid \mtrx{X}_{t},\mtrx{\theta}_{t}] = a_{t,m}/\sigma_{t}$. On the other hand, using Proposition \ref{prop:wlnw}, we obtain
\begin{align*}
\sum_{m=1}^{N}\mathbb{E}_{\pi}\left[w_{t,m}\ln w_{t,n}\mid \mtrx{X}_{t},\boldsymbol{\theta}_{t}\right]f_{t,m}
&=\sum_{\substack{m=1\\m\neq n}}^{N} \mathbb{E}_{\pi}\left[w_{t,m}\ln w_{t,n}\mid S_{t},\boldsymbol{\theta}_{t}\right]f_{t,m}+\mathbb{E}_{\pi}\left[w_{t,n}\ln w_{t,n}\mid S_{t},\boldsymbol{\theta}_{t}\right]f_{t,n}\\
&=\sum_{\substack{m=1\\m\neq n}}^{N}\frac{a_{t,m}}{\sigma_{t}}\left(\digamma\left(a_{t,n}\right)-\digamma\left(\sigma_{t}\right)-\frac{1}{\sigma_{t}}\right)f_{t,m}\\
&\quad+\frac{a_{t,n}}{\sigma_{t}}\left(\digamma\left(a_{t,n}\right)+\frac{1}{a_{t,n}}-\digamma\left(\sigma_{t}\right)-\frac{1}{\sigma_{t}}\right)f_{t,n}\\
&=\left(\digamma\left(a_{t,n}\right)-\digamma\left(\sigma_{t}\right)-\frac{1}{\sigma_{t}}\right)\sum_{m=1}^{N}\frac{a_{t,m}}{\sigma_{t}}f_{t,m}+\frac{1}{\sigma_{t}}f_{t,n}.
\end{align*}
This yields
\begin{align*}
\nabla J\left(\boldsymbol{\theta}_{t}\right)
&=\sum_{n=1}^{N}\left(\digamma\left(\sigma_{t}\right)-\digamma\left(a_{t,n}\right)\right)\left(\sum_{m=1}^{N}\frac{a_{t,m}}{\sigma_{t}}f_{t,m}\right)\nabla a_{t,n}\\
&\quad+\sum_{n=1}^{N}\left(\left(\digamma\left(a_{t,n}\right)-\digamma\left(\sigma_{t}\right)-\frac{1}{\sigma_{t}}\right)\left(\sum_{m=1}^{N}\frac{a_{t,m}}{\sigma_{t}}\right)f_{t,m}+\frac{1}{\sigma_{t}}f_{t,n}\right)\nabla a_{t,n}\\
&=\sum_{n=1}^{N}\left(f_{t,n}-\sum_{m=1}^{N}\frac{a_{t,m}}{\sigma_{t}}f_{t,m}\right)\frac{\nabla a_{t,n}}{\sigma_{t}}.
\end{align*}

Now note that
\[ \sum_{m=1}^{N}\frac{a_{t,m}}{\sigma_{t}}f_{t,m} 
= \mathbb{E}_{\pi} \left[ \mtrx{w}_{t}\mid \mtrx{X}_{t},\mtrx{\theta}_{t} \right]^\intercal f\left(\mtrx{X}_{t}\right) 
= \mathbb{E}_{\pi}\left[R_{t+1}\mid\mtrx{X}_{t},\mtrx{\theta}_{t}\right], \]
and that, by equation \eqref{eq:FactorModel}, $f_{t,n}$ is the expected return of asset $n$ between $t$ and $t+1$.

\subsection{Proof of lemma \ref{lem:ass}}
\label{lem:proof}

The Lagrange formulation of the problem,
\begin{align}
L(\bm{\theta})&=\bm{\theta}^\intercal\bm{X}^\intercal\bm{X}\bar{\bm{\beta}}-\frac{\gamma}{2}\bm{\theta}^\intercal\bm{X}^\intercal \E\left[(\bm{X}(\bm{\beta}-\bar{\bm{\beta}})+\bm{\epsilon})(\bm{\epsilon}+\bm{X}(\bm{\beta}-\bar{\bm{\beta}}))^\intercal \right]\bm{X}\bm{\theta} + \lambda(  \bm{\theta}^\intercal\bm{X}^\intercal\bm{1}_N-1) \label{eq:lagrange} \\
\frac{\partial L}{\partial \bm{\theta}}&=\bm{X}^\intercal\bm{X}\bar{\bm{\beta}}-\gamma\bm{X}^\intercal(\bm{X}\bm{\Sigma}_\beta\bm{X}^\intercal+\sigma_\epsilon^2\bm{I}_N)\bm{X}\bm{\theta}+\lambda \bm{X}^\intercal \bm{1}_N
\end{align}
leads, via the first order conditions, to the standard solution
\begin{align}
\bm{\theta}^*&=\gamma^{-1}(\bm{X}^\intercal(\bm{X}\bm{\Sigma}_\beta\bm{X}^\intercal+\sigma_\epsilon^2\bm{I}_N)\bm{X})^{-1}(\bm{X}^\intercal\bm{X}\bar{\bm{\beta}}+c\bm{X}^\intercal\bm{1}_N), \nonumber \\
&=\gamma^{-1}(\bm{\Sigma}_\beta\bm{X}^\intercal\bm{X}+\sigma_\epsilon^2\bm{I}_K)^{-1}(\bar{\bm{\beta}}+c(\bm{X}^\intercal\bm{X})^{-1}\bm{X}^\intercal\bm{1}_N) ,\label{eq:theta}
\end{align}
where $c$ is a constant which ensures that the budget constraint (to the right of Equation \ref{eq:maxprog}) is fulfilled. Note that $\bm{X}^\intercal\bm{X}$ is nonsingular because the characteristics are not redundant and because $N>K+1$. For the sake of completeness, we derive the expressions for the first inverse matrice below.

From (\ref{eq:XX}) and the definition of $ \bar{\bm{x}}$, it holds that
\begin{equation}
\bm{X}^\intercal\bm{X}=N(\hat{\bm{\Sigma}}_X+ \bar{\bm{x}}\bar{\bm{x}}^\intercal)=N(\text{diag}(\bm{\sigma}^2_X)+ \bar{\bm{x}}\bar{\bm{x}}^\intercal),
\label{eq:XX2}
\end{equation}
so that by the Sherman-Morrison formula, and because $\bm{\Sigma}_\beta=\text{diag}(\bm{\sigma}^2_\beta)$,
\begin{align}
(\bm{\Sigma}_\beta\bm{X}^\intercal\bm{X}+\sigma_\epsilon^2\bm{I}_K)^{-1} &= \left( N\text{diag}(\bm{\sigma}^2_X)\text{diag}(\bm{\sigma}^2_\beta)+N\text{diag}(\bm{\sigma}^2_\beta) \bar{\bm{x}}\bar{\bm{x}}^\intercal+\sigma_\epsilon^2 \bm{I}_K\right)^{-1}  \nonumber\\
& = N^{-1}\text{diag}(\bm{\sigma}^2)^{-1}\left(\bm{I}_K- \frac{\text{diag}(\bm{\sigma}^2_\beta) \bar{\bm{x}}\bar{\bm{x}}^\intercal \text{diag}(\bm{\sigma}^2)^{-1}}{1+\bar{\bm{x}}^\intercal\text{diag}(\bm{\sigma}^2)^{-1}\text{diag}(\bm{\sigma}^2_\beta)\bar{\bm{x}}}\right), \label{eq:simp}
\end{align}
where $\text{diag}(\bm{\sigma}^2)=\text{diag}(\bm{\sigma}^2_X)\text{diag}(\bm{\sigma}^2_\beta)+N^{-1}\sigma_\epsilon^2 \bm{I}_K$ - this form being a strong echo of the structure in Equation (\ref{eq:model}). This proves the first point.

Now, let us make the extreme simplification, as in our empirical section, that $\bar{\bm{x}}^\intercal=[1 \quad \bm{0}_K^\intercal]$, so that firm characteristics have zero sample mean - apart for the first one. This is not uncommon in the recent literature as long as the data is preprocessed (see \cite{freyberger2020dissecting}, \cite{gu2020empirical}, \cite{kelly2019characteristics} and \cite{koijen2019demand}). Then, $\bm{X}^\intercal\bm{X}= N\text{diag}(\tilde{\bm{\sigma}}_{X}^2)$, where the modified vector of variances $\tilde{\bm{\sigma}}_{X}^2$ is simply $ \bm{\sigma}_{X}^2$ with the first element equal to one (instead of zero). The ratio in \eqref{eq:simp} vanishes (because $\sigma^2_{\beta,1}=0$ and $ \bar{\bm{x}}\bar{\bm{x}}^\intercal$ is empty apart for its unit first element) and   
\begin{align*}
\bm{\theta}_* =(N\gamma)^{-1}  \text{diag}(\bm{\sigma}^2)^{-1} (\bar{\bm{\beta}}+c \, \text{diag}(\tilde{\bm{\sigma}}_{X}^2)^{-1}\bar{\bm{x}}),
\end{align*}
from which the lower part of \eqref{eq:thetastar} is derived (the constant $c$ impacts only $\theta^{(0)}_*$). In this case, the budget constraint is only binding for the first asset. Indeed, because $\bm{1}_N \bm{X}=\bar{\bm{x}}=[1 \quad \bm{0}_K^\intercal]^\intercal$, the sum of weights linked to the non-constant factors is always equal to zero. Thus, $\theta^{(0)}_*$, which is linked to a constant column, must satisfy $\theta^{(0)}_*N=1$.


\section{Dirichlet distributions and portfolios in high dimensions}
\label{sec:a}
One major issue with the Dirichlet distribution is the computation of the scaling constant in high dimension. More precisely, let us consider the log of this quantity:
\begin{equation}
c=\log(B(\bm{a}))=\sum_{n=1}^N\log(\Gamma(a_n))-\log\left(\Gamma\left(\sum_{n=1}^Na_n\right)\right).
\label{eq:clog}
\end{equation}
When $N$ is large and the $a_n$ are free, both terms can reach levels that are beyond what machines can handle when exponentiated. Thus, we need to impose restrictions. We do it in two steps. First, we set some lower and upper bound on the $a_n$. In a second stage, we compute an upper value for $N$ that will depend on the range of the $a_n$. This second step is the most technical and we provide the details below. The third and last step is to determine a tradeoff. 

Before we continue, we recall that the $a_n$ dictate the allocation of the agent and that, on average, the position in asset $n$ is equal to $a_n\left(\sum_{n=1}^N a_n \right)^{-1}$. For obvious risk-management reasons, it is preferable to diversify portfolios. In our framework, we assume that there exists a constant $\delta >1$ such that:
\begin{equation}
\frac{1}{\delta N} \le a_n\left(\sum_{n=1}^N a_n \right)^{-1} \le \frac{\delta}{N}, \quad n=1,\dots, N.
\label{eq:div}
\end{equation}
In practice, the minimum value of $\delta$ will be driven by the data, and we discuss realistic ranges below. This constraint helps measure if the portfolio is on average well balanced and does not make extreme bets. The smaller $\delta$ is, the higher the diversification of the positions. Notably, under condition (\ref{eq:div}), the mean of the $a_n$, $m_a$, is such that 
\begin{equation}
\delta^{-1}a_+ \le  m_a\le \delta a_-, \text{ with } a_+=\underset{n}{\text{max}} \ a_n, \quad a_-=\underset{n}{\text{min}} \ a_n,.
\label{eq:div2}
\end{equation}


To further explicit our idea, we fix a maximum threshold $\kappa_{\text{max}}$ beyond which we consider that the two terms in Equation (\ref{eq:clog})  have numerically exploded. The two terms in Equation (\ref{eq:clog}) have very different asymptotics when the $a_n$ are large or small, hence the treatment is not symmetric. We start with problems when the $a_n$ are large. Given the strong convexity of the $\Gamma$ function, this is more impactful for the second term in (\ref{eq:clog}). We seek an upper bound $a_+$ for the $a_n$ such that this second term remains below $\kappa_{\text{max}}$, i.e.,
$$\log\left(\Gamma\left(\sum_{n=1}^N a_n\right)\right)\le \kappa_{\text{max}}. $$
Although the inverse of the $\Gamma$ function exists (at least when its argument is large enough, see \cite{uchiyama2012principal}), it is not straightforward to compute. We thus resort to Stirling's formula instead and seek to simplify
\begin{equation}
\log\left(\sqrt{2\pi\left(\sum_{n=1}^N a_n-1\right)}\left(\frac{\sum_{n=1}^N a_n-1}{e}\right)^{\sum_{n=1}^N a_n-1} \right) \le \kappa_{\text{max}}.
\label{eq:stir}
\end{equation}
If we omit the first negligible term inside the square root, this is equivalent to 
$$\left(\sum_{n=1}^N a_n-1\right)\left(\log\left(\sum_{n=1}^N a_n-1\right)-1\right)\le \kappa_{\text{max}}.$$
 As a first order (rough) approximation, we reduce this expression to 
 $$\left(\sum_{n=1}^N a_n\right)\log\left(\sum_{n=1}^N a_n\right)\le  \kappa_{\text{max}},$$ 
 that is, $\sum_{n=1}^N a_n\le \frac{\kappa_{\text{max}}}{W(\kappa_{\text{max}})}\sim  \frac{\kappa_{\text{max}}}{\log(\kappa_{\text{max}})}$, where $W$ is the principal branch of the Lambert function. Its asymptotic behavior for large arguments is indeed $W(z)\sim \log(z)$ (see Section 4.13 in \cite{olver2010nist}). Given (\ref{eq:div2}), a rule of thumb constraint that links $N$ and $a_+$ is 
\begin{equation}
a_+ \le \frac{\delta}{ N} \frac{ \kappa_{\text{max}}}{ \log( \kappa_{\text{max}})} \Leftrightarrow N \le \frac{\delta \kappa_{\text{max}}}{a_+ \log( \kappa_{\text{max}})},
\label{eq:condp}
\end{equation}
where we purposefully underline that the condition can also be viewed as a limit on the number of assets. 

The first term in \ref{eq:clog} relates to the lower bound on the $a_n$. Indeed, as $z$ shrinks to zero, $\Gamma(z)$ is equivalent to $z^{-1}$. Thus, if the $a_n$ are small and $a_-$ is sufficiently close to zero,
\begin{equation}
\sum_{n=1}^N\log(\Gamma(a_n))\le N\log\left(\frac{1}{a_-}\right) \le   \kappa_{\text{max}} \quad \Leftrightarrow \quad N \le  \kappa_{\text{max}}/\log(a_-^{-1}) \quad \Leftrightarrow \quad a_-\ge e^{- \kappa_{\text{max}}/N} . 
\label{eq:condm}
\end{equation}

Conditions (\ref{eq:condp}) and (\ref{eq:condm}) link the bounds of the $a_n$ to the number of assets $N$. In Figure \ref{fig:plot}, we illustrate them by assigning values to $ \kappa_{\text{max}}$, $\delta$ and $N$. Taking $\kappa_{\text{max}}=100$ allows $B(\bm{a})$ to range from $e^{-100}$ to $e^{100}$, which is a large magnitude. In the figure, as the number of assets increases, the range of the $a_n$ shrinks. 

For our empirical study, we pick $a_-=0.02$ and $a_+=1.6$. These values are optimal empirically because we obtain errors outside this range.

\begin{figure}[!h]
\begin{center}
\includegraphics[width=10cm]{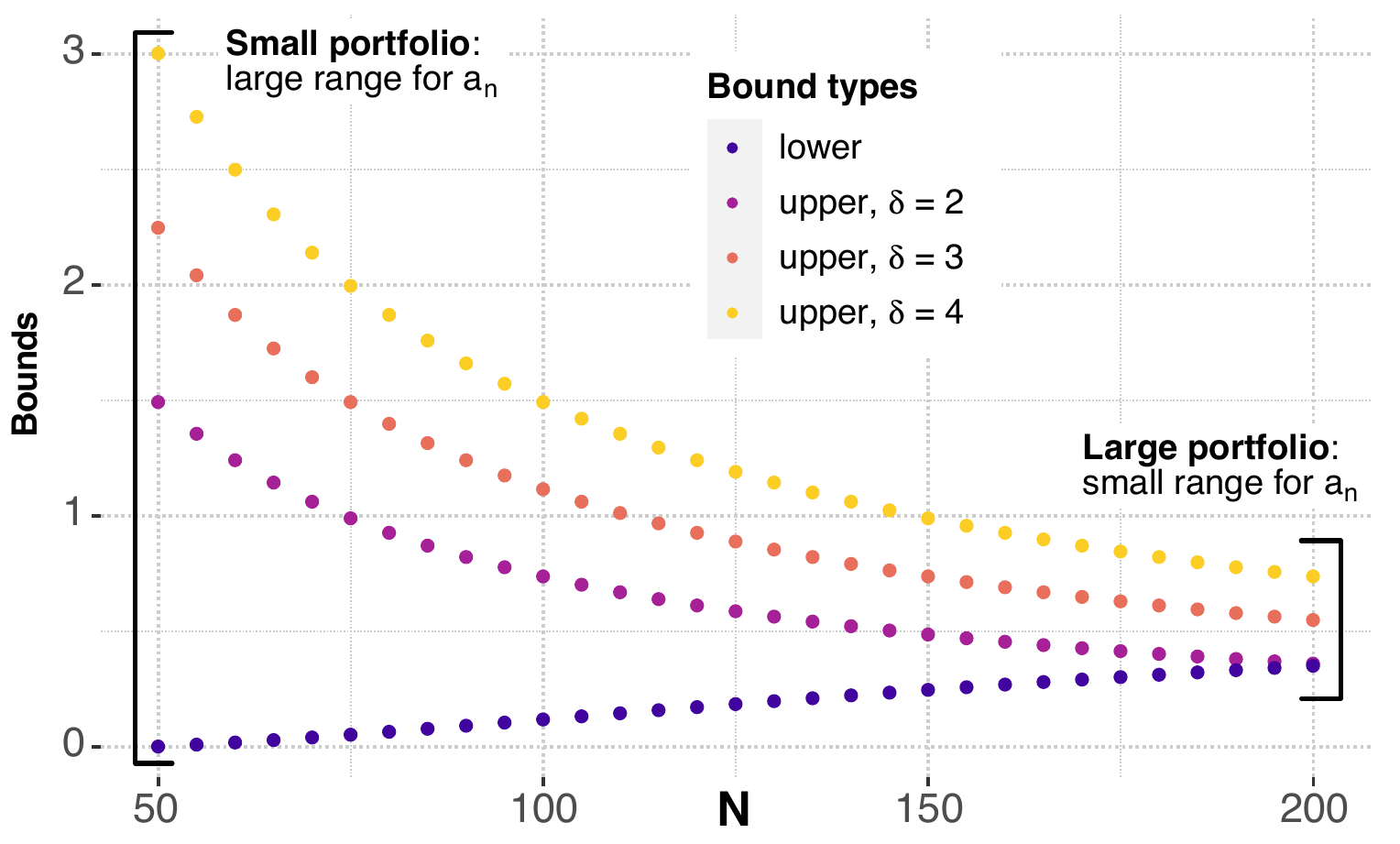}\vspace{-3mm}
\caption{\textbf{Intervals for the $a_n$}. \small We show the lower ($a_-$) and upper ($a_+$) bound for the $a_n$ when the number of assets is fixed to 50, 100 or 200 and $ \kappa_{\text{max}}=100$. They are derived from Equations (\ref{eq:condp}) and (\ref{eq:condm}). The black line is the $\Gamma$ function.}
\label{fig:plot}
\end{center}
\end{figure}


\bibliographystyle{chicago}
\bibliography{bib}
 
\end{appendices}

\end{document}